\documentclass[11pt,letterpaper]{article}
\usepackage[utf8]{inputenc}
\usepackage[margin=1.in]{geometry}
\usepackage{amsthm}
\usepackage{xcolor}
\usepackage{mathtools,amsmath,amssymb}
\usepackage{physics, graphicx}
\usepackage{bm}
\usepackage[ruled,vlined]{algorithm2e}
\usepackage{booktabs}
\usepackage{caption}
\usepackage{rotating}
\usepackage{enumitem}
\usepackage{tikz}
\usetikzlibrary{calc,positioning,arrows.meta}
\usepackage{authblk}
\usepackage[colorlinks=true, allcolors=blue]{hyperref}
\usepackage[capitalise,nameinlink]{cleveref}
\usepackage[style=alphabetic, backend=biber, minalphanames=3, maxalphanames=5, maxbibnames=99, sorting=nyt]{biblatex}

\setlist[itemize]{topsep=3pt}
\setlist[enumerate]{topsep=3pt}

\usepackage{thmtools}

\theoremstyle{plain}
\newtheorem{theorem}{Theorem}
\newtheorem{proposition}[theorem]{Proposition}
\newtheorem{lemma}[theorem]{Lemma}

\usepackage[compact]{titlesec}
\titlespacing{\subsection}{0pt}{1.5ex}{0ex}
\titlespacing{\subsubsection}{0pt}{1ex}{0ex}
\titlespacing{\paragraph}{0pt}{1.5ex}{1ex}
\title{Provable Classical and Quantum Local Algorithms for Max-$k$-Cut \\ and Quantum Advantage at Moderate Girth}
\author[1]{Anuj Apte}
\author[1]{Abid Khan}
\author[2]{Edward Farhi}
\author[2]{Kunal Marwaha}
\author[1]{Sami Boulebnane}
\author[1]{Ruslan~Shaydulin}
\affil[1]{Global Technology Applied Research, JPMorganChase, New York, NY 10017 USA}
\affil[2]{Google Quantum AI, Venice, CA 90291 USA}
\date{\vspace{-5.5em}}
\begin{document}
\maketitle
\begin{abstract}
Broadening the study of quantum optimization algorithms from binary to $k$-element alphabets has been shown to open new avenues for potential quantum advantage.
A quantum advantage claim for approximate optimization requires showing that, under the same assumptions, an efficient quantum algorithm provably achieves a better performance than can be proven for the best known efficient classical algorithms.
We study local classical and quantum algorithms for Max-$k$-Cut on $d$-regular graphs of girth $g$.
We advance classical algorithms for this problem by developing a local vector algorithm based on the explicit vector construction of Thompson, Parekh, and Marwaha (TPM). 
Our algorithm gives the best provable cut fraction guarantee among known efficient classical algorithms on regular graphs for $k\geq3$. To evaluate the performance of QAOA under identical assumptions of girth and regularity, we develop tensor network techniques for general $d$ and $k \geq 2$. In addition, using an equivalence to a coupled qudit--boson system, we compute the QAOA performance in the infinite-degree limit.
Together, these techniques give provable guarantees on QAOA performance on large graphs. 
Despite the improvements we introduce to the classical algorithm, QAOA achieves a better cut fraction guarantee for depths $p\geq 9$, corresponding to girth $g\geq 20$, for both finite- and infinite-degree regimes. Thus, we obtain an apparent quantum advantage from applying QAOA to the Max-$k$-Cut problem.

\end{abstract}
\vspace{1em}

\section{Introduction}
\label{sec:introduction}

  Thirty years ago we saw that quantum computers can have computational advantage. The foremost examples of this are Shor's algorithm for factoring \cite{Shor1994, Shor1997} and Grover's algorithm for unstructured search \cite{Grover1996} in an oracle setting. Apart from finding exact solutions we can ask if quantum algorithms can provide good solutions to combinatorial search problems with a reasonable runtime (i.e.\ allowed to scale as a small polynomial in the problem size). For the Max-Cut problem on 3-regular graphs, Farhi et al.~\cite{lower_bounding_max_cut_qaoa} show that the Quantum Approximate Optimization Algorithm (QAOA) at polynomial depth achieves better cut fractions than the best known classical algorithm with a provable performance guarantee, provided that the girth of the graph $g$ equals or exceeds $16$. 

  When making a comparison between quantum and classical algorithms, it is important to understand what is assumed about the structure of the problem under study. In the case of QAOA for graph problems, the analysis of its performance is dramatically simplified if one restricts to regular graphs with a girth $g \geq 2p+2$ where $p$ is the QAOA depth, i.e.\ the number of alternating operator layers applied. Each layer consists of one application of the problem (cost) Hamiltonian and one of the mixer Hamiltonian, parameterized by $2p$ continuous angles. The essential reason for this simplification is that the expectation value of an observable defined on an edge can be computed solely by restricting to a tree subgraph of the original graph. Note that many classical local algorithms and QAOA are affected by common structural obstructions, including the Overlap Gap Property~\cite{qaoa_needs_whole_graph_typical_case,limitations_local_quantum_algorithms}. The key question then is whether one can come up with better classical algorithms using exactly the same structural assumption about the graphs, and whether QAOA can outperform these algorithms. 

   Max-Cut is the computational problem which asks one to partition the vertices of a graph into two subsets to maximize the size of the ``cut'', i.e.\ the number of edges that cross between the subsets. Determining the maximum cut of an arbitrary graph is NP-hard \cite{Karp1975OnTC}, and it remains NP-hard when restricted to regular graphs \cite{Yannakakis1978}. For this problem, classical algorithms were developed exploiting the locally tree-like structure \cite{hirvonen2014largecutslocalalgorithms, thompson2022explicit, hastings2021classicalalgorithmbeatsfrac12frac2pifrac1sqrtd}, and in tandem numerical and analytical techniques were developed to analyze QAOA \cite{qaoa_maxcut_high_depth, lower_bounding_max_cut_qaoa}. The outcome of this process was a concrete identification of a regime in which QAOA exceeds the performance of the best performing classical algorithm as mentioned above. In this paper we extend this program to Max-$k$-Cut, which asks for a partition of the vertices into $k$ subsets that maximizes the number of edges whose endpoints lie in different subsets; Max-Cut on regular graphs is the special case $k=2$. We do so by simultaneously developing a state-of-the-art classical algorithm for regular graphs and devising new methods to analyze the performance of QAOA. 

    We measure performance of the algorithms by the \emph{cut fraction}, namely the number of cut edges divided by total number of edges $|E|$. A uniformly random assignment cuts each edge with probability $(k-1)/k$, and so achieves expected cut fraction $(k-1)/k$. This should be distinguished from an \emph{approximation ratio}, which compares the value returned by an algorithm to the optimum value $\mathsf{OPT}_k(G)$. The distinction between cut fraction and approximation ratio is important. Algorithms such as the Frieze--Jerrum semidefinite program \cite{Frieze1997}, which generalizes the Goemans--Williamson algorithm \cite{goemans1995}, give strong approximation-ratio guarantees that hold for all graphs. However, when the optimum cut fraction itself is close to $(k-1)/k$, the cut fraction guaranteed by such a multiplicative guarantee may be below the random-assignment baseline. For random $d$-regular graphs, this occurs above the $k$-colorability threshold \cite{kemkes2010chromatic}. We prove in \cref{prop:fj_below_baseline} that there exists $d_k^{\star}$ such that for every degree $d > d_k^{\star}$ we get $\alpha^{\mathrm{FJ}}_k \mathsf{OPT}_k(G)/|E| < (k-1)/k$, where $\alpha^{\mathrm{FJ}}_k$ is the Frieze--Jerrum approximation ratio.

In contrast to the semidefinite programming algorithms which work for all graphs, a classical algorithm that exploits the structure of regular graphs with a given girth is the explicit vector algorithm of Thompson, Parekh, and Marwaha (TPM) for Max-2-Cut \cite{thompson2022explicit}. Rather than solving an SDP, the algorithm assigns unit vectors to vertices according to their distance from a root and rounds them with random directions, with the optimal vectors determined by the minimum eigenvector of a small tridiagonal matrix whose size is controlled by the girth. On any $d$-regular graph of girth $g \geq 4$ it produces an assignment whose expected cut fraction provably exceeds the random baseline. Asymptotically, as $d \to \infty$ the cut fraction takes the form $1/2 + \Theta(1/\sqrt{d})$.

   Our focus in this work is on approximation algorithms with provable guarantees that hold for every instance of the chosen ensemble, including worst-case inputs. This differs from the heuristic or average-case setting, where methods such as simulated annealing are evaluated empirically on random instances. Strong empirical performance on typical instances does not imply any worst-case guarantee. Thus, we do not benchmark against all practical heuristics, but analyze only those classical and quantum local algorithms whose cut-fraction guarantees can be proven rigorously. 

\subsection{Results}

There has been much activity in developing provable local algorithms for graph optimization and in analyzing QAOA on regular graphs~\cite{hirvonen2014largecutslocalalgorithms, thompson2022explicit, hastings2021classicalalgorithmbeatsfrac12frac2pifrac1sqrtd, qaoa_maxcut_high_depth, lower_bounding_max_cut_qaoa, apte2026quantumapproximateoptimizationinteger}. We contribute to this research thread by developing a state-of-the-art classical algorithm for Max-$k$-Cut on regular graphs and new tools for analyzing QAOA on the same class of graphs, allowing a concrete head-to-head comparison of the two:

\begin{enumerate}
\item
\emph{A classical Local Vector algorithm that provably beats the random baseline} (\cref{sec:vector_algorithm}).

We generalize the explicit vector algorithm of Thompson, Parekh, and Marwaha (TPM) from Max-2-Cut to Max-$k$-Cut. Our algorithm attains a cut fraction exceeding the random baseline of $(k-1)/k$ for all regular graphs with girth $g \geq 4$. Furthermore, the improvement upon the random baseline is asymptotically $\Theta(1/\sqrt{d})$ as $d \to \infty$. We further improve the algorithm with a one-step rounding rule, in which each vertex discounts its own label preference by non-positive messages from its neighbors, and find numerically that this strictly increases the leading $1/\sqrt{d}$ coefficient for every $k$ and every girth $g \geq 6$. Because this improvement is additive and holds exactly where the Frieze--Jerrum guarantee is vacuous, it provides the best provable cut-fraction guarantee among efficient classical algorithms on regular graphs for $k \geq 3$. 

\item
\emph{A qudit--boson reformulation of QAOA in the infinite-degree limit} (\cref{subsec:qaoa_asymptotics,subsec:qudit_boson}).

Prior work analyzed qudit QAOA for Max-$k$-Cut only at finite degree $d$, through exact tree tensor-network contraction \cite{apte2026quantumapproximateoptimizationinteger}, which does not by itself reveal the large-$d$ behavior. In this work, we take the $d \to \infty$ limit and express the resulting cut fraction in the same $(k-1)/k + \Theta(1/\sqrt d)$ normalization as the classical algorithm, putting the two on a common footing.
In this infinite-degree limit, we further express QAOA expectation values in terms of a qudit coupled to $(k-1)\times(p+1)$ bosonic modes, extending the qubit results of~\cite{2w94-rymn}.

\item
\emph{A degree-independent tree tensor-network algorithm for qudit QAOA} (\cref{subsec:qaoa_tensor_algorithm}).

We generalize the exact tree tensor-network contraction of~\cite{lower_bounding_max_cut_qaoa} from qubits to $k$-level qudits, computing the depth-$p$ Max-$k$-Cut QAOA cut fraction on $d$-regular graphs of girth $g \geq 2p+2$. Collapsing the identical branches of the light-cone tree by entrywise $(d-1)$th powers makes the contraction cost $\mathcal{O}(p\,k^{2(p+1)})$ in time and $\mathcal{O}(k^{2p})$ in memory, both independent of the degree $d$, which lets us reach the depths and degrees needed for a meaningful head-to-head comparison. Putting the classical and quantum analyses side by side across the parameters $(k, d, p)$, the Local Vector algorithm is overtaken by QAOA once $p \geq 9$, i.e.\ $g \geq 20$, identifying a regime for quantum advantage on graph optimization at moderate girth. The head-to-head comparison across $(k, d, p)$ is presented in \cref{sec:comparison}. In this work we focus on $k=2,3,4$ due to the exponential growth in complexity of evaluating QAOA expectation values with $k$. 

\end{enumerate}

\subsection{Organization}

The rest of the paper is organized as follows. \Cref{sec:preliminaries} introduces the Max-$k$-Cut problem and the Frieze--Jerrum SDP, and explains why its worst-case guarantee becomes vacuous on graphs where the optimal cut fraction is near the random baseline. \Cref{sec:vector_algorithm} presents the classical Local Vector algorithm, including the neighbor-message rounding rule and pseudocode, and gives the full cut-fraction analysis in both the finite-degree and infinite-degree regimes. \Cref{sec:qaoa} then develops QAOA for Max-$k$-Cut: after setting up the algorithm, we analyze the finite-degree cut fraction via tree tensor-network contraction and take the infinite-degree limit by generalizing techniques from~\cite{qaoa_maxcut_high_depth}; we then introduce the qudit--boson machinery enabling approximate evaluation of QAOA energies at larger depth. \Cref{sec:comparison} places the classical and quantum algorithms side by side across $(k, d, p)$. We conclude in \cref{sec:conclusion} with a discussion of the Potts--Parisi upper bound and other problems for future work. 

\section{Preliminaries}
\label{sec:preliminaries}

\subsection{The Max-$k$-Cut Problem}
\label{subsec:maxkcut}

Given a graph $G = (V, E)$ with $n = |V|$ vertices and $|E|$ edges, the Max-$k$-Cut problem asks for a partition of $V$ into $k$ disjoint subsets $S_1, \ldots, S_k$ that maximizes the number of edges with endpoints in different subsets (see \cref{fig:maxkcut_example}).
Equivalently, we seek an assignment $\chi: V \to \mathbb{Z}_k = \{0, 1, \ldots, k-1\}$ that maximizes the cost function
\begin{equation}
    C_G(\chi) = \sum_{\{u,v\} \in E} \mathbf{1}[\chi(u) \neq \chi(v)],
\end{equation}
where $\mathbf{1}$ is the indicator function which only fires if its argument is true.

\begin{figure}[!htbp]
    \centering
    \begin{tikzpicture}[scale=0.85,
        vertex/.style={circle, draw, thick, minimum size=18pt, font=\small},
        cut edge/.style={dashed, thick, gray},
        uncut edge/.style={thick, black},
        color1/.style={fill=red!40},
        color2/.style={fill=blue!40},
        color3/.style={fill=green!40},
    ]

\node[vertex, color1] (v1) at (0, 2) {$0$};
\node[vertex, color2] (v2) at (2, 2) {$1$};
\node[vertex, color3] (v3) at (4, 2) {$2$};
\node[vertex, color2] (v4) at (5, 0.5) {$1$};
\node[vertex, color1] (v5) at (3, -0.5) {$0$};
\node[vertex, color3] (v6) at (1, 0) {$2$};
\node[vertex, color1] (v7) at (-1, 0.5) {$0$};

\draw[cut edge] (v1) -- (v2);
\draw[cut edge] (v2) -- (v3);
\draw[cut edge] (v1) -- (v6);
\draw[cut edge] (v2) -- (v6);
\draw[cut edge] (v3) -- (v4);
\draw[cut edge] (v4) -- (v5);
\draw[cut edge] (v5) -- (v6);
\draw[cut edge] (v6) -- (v7);
\draw[cut edge] (v2) -- (v5);

\draw[uncut edge] (v1) -- (v7);
\draw[uncut edge] (v1) -- (v5);

\draw[cut edge] (0, -1.7) -- (1, -1.7);
\node[right] at (1, -1.7) {\small Cut edge};

\draw[uncut edge] (3, -1.7) -- (4, -1.7);
\node[right] at (4, -1.7) {\small Uncut edge};

\end{tikzpicture}
    \caption{Example of Max-3-Cut on a graph with 7 vertices and 11 edges. Vertices are assigned to subsets with labels $\{0, 1, 2\}$. Dashed edges connect vertices belonging to different sets (cut), while solid edges connect vertices of the same set (uncut).}
    \label{fig:maxkcut_example}
\end{figure}
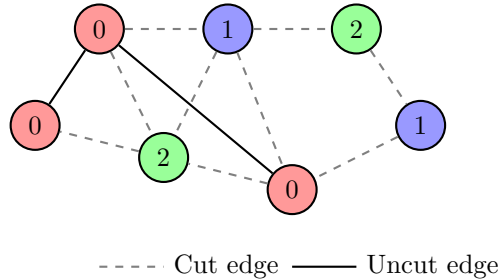

The \emph{cut fraction} of an assignment $\chi$ is
\begin{equation}
    \text{cut fraction}(G, \chi) = \frac{C_G(\chi)}{|E|},
\end{equation}
and the \emph{approximation ratio} is the ratio to the optimal cut value $\mathsf{OPT}_k(G) = \max_\chi C_G(\chi)$.
Since $\mathsf{OPT}_k(G) \leq |E|$, the cut fraction provides a lower bound on the approximation ratio.

For a uniformly random labeling, each edge is cut with probability $(k-1)/k$, giving an expected cut fraction of
\begin{equation}
    \mathbb{E}_\chi[\text{cut fraction}] = \frac{k-1}{k}.
\end{equation}
This value serves as a natural baseline; any useful approximation algorithm should exceed it.

The \emph{girth} of a graph is the length of its shortest cycle; this structure enables simplified analysis of both the vector algorithm and QAOA. Throughout this work, we parametrize girth as $g = 2m$ 
for the vector algorithm and $g = 2p + 2$ for QAOA at depth $p$. In this parametrization, the $(m - 1)$-neighborhood (respectively, the $p$-neighborhood) of any edge is  tree-like.

\subsection{Semidefinite Programming Relaxation}
\label{subsec:sdp}

The Max-$k$-Cut problem admits a semidefinite programming (SDP) relaxation that provides the best known worst-case approximation guarantees. We briefly review the Frieze--Jerrum algorithm \cite{Frieze1997}, which generalizes the Goemans--Williamson approach for Max-2-Cut. Beyond its role as a benchmark, the Frieze--Jerrum algorithm is the conceptual template for our own construction: we borrow its two central ideas: \emph{(i)} assign a vector to each vertex, and \emph{(ii)} round the vectors with random directions to obtain a label, while discarding the SDP itself in favor of an explicit, girth-aware vector assignment.

For $k = 2$, the Goemans--Williamson algorithm \cite{goemans1995} solves

\begin{align}
    \text{Maximize} \quad & \frac{1}{2} \sum_{\{u,v\} \in E} (1 - \mathbf{v}_u \cdot \mathbf{v}_v) \nonumber \\
    \text{Subject to} \quad & \mathbf{v}_i \in S^{n-1}, \; \forall i \in V,
\end{align}
where $S^{n-1}$ is the unit sphere in $\mathbb{R}^n$, and rounds the solution using a random hyperplane to obtain a cut with approximation ratio at least $\alpha_{\text{GW}} \approx 0.878$.

For general $k$, the key idea is to embed the $k$ labels as vertices of a regular simplex in $\mathbb{R}^{k-1}$. To construct this simplex, begin with the $k$ standard basis vectors $e_1, \ldots, e_k \in \mathbb{R}^k$, subtract the centroid $c = (1/k, \ldots, 1/k)$ from each, and normalize to obtain unit vectors $q_1, \ldots, q_k$ satisfying
\begin{equation}
    q_a \cdot q_b = \begin{cases}
        1 & \text{if } a = b, \\
        -1/(k-1) & \text{if } a \neq b.
    \end{cases}
\end{equation}
If we assign $q_{\chi(u)}$ to vertex $u$ for a labeling $\chi$, then the indicator for a cut edge becomes
\begin{equation}
    \mathbf{1}[\chi(u) \neq \chi(v)] = \frac{k-1}{k}(1 - q_{\chi(u)} \cdot q_{\chi(v)}),
\end{equation}
and the Max-$k$-Cut objective can be written as
\begin{equation}
    \frac{k-1}{k}\sum_{\{u,v\} \in E} (1 - q_{\chi(u)} \cdot q_{\chi(v)}).
\end{equation}

The SDP relaxation replaces the discrete assignment $q_{\chi(u)}$ with arbitrary unit vectors $\mathbf{y}_u \in \mathbb{R}^n$, and introduces the Gram matrix $X_{uv} = \mathbf{y}_u \cdot \mathbf{y}_v$. The off-diagonal constraint $X_{uv} \geq -1/(k-1)$ for $u \neq v$ ensures that the relaxed solution does not exploit correlations impossible in any valid $k$-cut. This yields the SDP:
\begin{align}
    \text{Maximize} \quad & \frac{k-1}{k}\sum_{\{u,v\} \in E} (1 - X_{uv}) \nonumber \\
    \text{Subject to} \quad & X \succeq 0, \quad X_{ii} = 1, \quad X_{uv} \geq -\frac{1}{k-1} \; (u \neq v).
\end{align}
After solving, one factors the optimal Gram matrix $X^* = UU^T$ to obtain vertex vectors $\mathbf{y}_v$ as rows of $U$. This is the first idea we borrow: \emph{every vertex is assigned a vector}.

The randomized rounding procedure samples $k$ independent random unit vectors $\mathbf{g}_1, \ldots, \mathbf{g}_k$ in the embedding space and assigns label $a$ to each vertex $u$ by maximizing the inner product $\mathbf{y}_u \cdot \mathbf{g}_a$. This generalizes hyperplane rounding from $k=2$ to arbitrary $k$, and is the second idea we borrow: \emph{a label is obtained by rounding the vector against random directions}. The worst-case approximation ratio $\alpha^{\mathrm{FJ}}_k$ satisfies $\alpha^{\mathrm{FJ}}_k > (k-1)/k$ for all $k$, with $\alpha^{\mathrm{FJ}}_k \approx 1 - 1/k + 2\ln(k)/k^2$ for large $k$ \cite{Frieze1997, de2004approximate}.

The distinction between approximation ratio and cut fraction is essential for understanding when explicit vector algorithms can provably outperform SDP-based methods.
The Frieze--Jerrum algorithm guarantees

\begin{equation}
    \mathsf{CUT} \geq \alpha^{\mathrm{FJ}}_k \cdot \mathsf{OPT}_k(G),
\end{equation}
where $\alpha^{\mathrm{FJ}}_k$ is a constant depending only on $k$. This is a \emph{worst-case} guarantee; and holds for every graph.

However, this does not imply a bound on the cut fraction unless we know $\mathsf{OPT}_k(G)/|E|$.

For Max-2-Cut, Dembo, Montanari, and Sen \cite{Dembo_2017} proved that on random $d$-regular graphs,

\begin{equation}
    \frac{\mathsf{OPT}_2(G)}{|E|} \to \frac{1}{2} + \frac{P_*}{\sqrt{d}} \quad \text{as } n ~, ~ d \to \infty,
\end{equation}
where $P_* \approx 0.7632$ is the Parisi constant \cite{Parisi1979, Talagrand2006}. Thus the Goemans--Williamson guarantee gives
\begin{equation}
    \text{cut fraction} \geq \alpha_{\text{GW}} \cdot \left(\frac{1}{2} + \frac{P_*}{\sqrt{d}}\right) \approx 0.439 + O(1/\sqrt{d}),
\end{equation}
which is \emph{less than} the random baseline of $1/2$ for large $d$. In contrast, the explicit vector algorithm of Thompson, Parekh, and Marwaha \cite{thompson2022explicit} achieves cut fraction $1/2 + c/\sqrt{d}$ with $c > 0$, providing a \emph{provable} improvement over random.

For Max-$k$-Cut, the situation depends on the chromatic number. Kemkes, Pérez-Giménez, and Wormald \cite{kemkes2010chromatic} showed that random $d$-regular graphs are $k$-colorable with high probability (implying $\mathsf{OPT}_k = |E|$ with high probability) when $d \lesssim 2(k-1)\ln(k-1)$. Above this threshold, the optimal cut fraction is less than 1 and approaches $(k-1)/k$ as $d$ grows. We now make quantitative the sense in which the Frieze--Jerrum \emph{guarantee} degrades in this regime. The following proposition uses only the near-Ramanujan spectral bound on random regular graphs~\cite{friedman2004proofalonssecondeigenvalue} and requires no analog of the Dembo--Montanari--Sen constant for $k > 2$.

\begin{proposition}[The Frieze--Jerrum guarantee falls below the random baseline]
\label{prop:fj_below_baseline}
Let $\alpha^{\mathrm{FJ}}_k \in (0, 1)$ be the Frieze--Jerrum approximation ratio, and define the explicit degree threshold
\begin{equation}
    d_k^{\star} \;:=\; \frac{4\,(\alpha^{\mathrm{FJ}}_k)^2}{(1 - \alpha^{\mathrm{FJ}}_k)^2}.
    \label{eq:dk_star}
\end{equation}
Then for every $d$-regular graph $G$ whose adjacency matrix satisfies the near-Ramanujan bound $\lambda_{\min}(A_G) \geq -2\sqrt{d-1}$ (which holds asymptotically almost surely for random $d$-regular graphs by Friedman's theorem~\cite{friedman2004proofalonssecondeigenvalue}), and for every degree $d > d_k^{\star}$, the cut fraction certified by the Frieze--Jerrum guarantee is strictly less than the random baseline:
\begin{equation}
    \alpha^{\mathrm{FJ}}_k \cdot \frac{\mathsf{OPT}_k(G)}{|E|} \;<\; \frac{k-1}{k}.
\end{equation}
\end{proposition}

\begin{proof}
Assigning the simplex vectors $q_{\chi(u)}$ of \cref{subsec:sdp} to an optimal assignment and using $q_a \cdot q_b \geq -1/(k-1)$ together with the eigenvalue bound $\sum_{\{u,v\} \in E} q_{\chi(u)} \cdot q_{\chi(v)} \geq \tfrac{1}{2}\lambda_{\min}(A_G)\sum_u \|q_{\chi(u)}\|^2 = \tfrac{1}{2}\lambda_{\min}(A_G)\, n$, we bound the optimal cut. Since $|E| = dn/2$,
\begin{equation}
    \frac{\mathsf{OPT}_k(G)}{|E|}
    = \frac{k-1}{k}\cdot \frac{1}{|E|}\sum_{\{u,v\} \in E}\bigl(1 - q_{\chi(u)} \cdot q_{\chi(v)}\bigr)
    \leq \frac{k-1}{k}\left(1 - \frac{\lambda_{\min}(A_G)}{d}\right).
\end{equation}
Applying $\lambda_{\min}(A_G) \geq -2\sqrt{d-1}$ and then $2\sqrt{d-1}/d \leq 2/\sqrt d$ gives
\begin{equation}
    \frac{\mathsf{OPT}_k(G)}{|E|} \leq \frac{k-1}{k}\left(1 + \frac{2\sqrt{d-1}}{d}\right) \leq \frac{k-1}{k}\left(1 + \frac{2}{\sqrt d}\right).
    \label{eq:opt_upper}
\end{equation}
Hence the Frieze--Jerrum certified cut fraction satisfies
\begin{equation}
    \alpha^{\mathrm{FJ}}_k \cdot \frac{\mathsf{OPT}_k(G)}{|E|} \leq \alpha^{\mathrm{FJ}}_k\,\frac{k-1}{k}\left(1 + \frac{2}{\sqrt d}\right).
\end{equation}
This upper bound is strictly below $(k-1)/k$ precisely when $\alpha^{\mathrm{FJ}}_k\bigl(1 + 2/\sqrt d\bigr) < 1$, i.e.\ when $2/\sqrt d < (1 - \alpha^{\mathrm{FJ}}_k)/\alpha^{\mathrm{FJ}}_k$, which rearranges to $d > 4(\alpha^{\mathrm{FJ}}_k)^2/(1-\alpha^{\mathrm{FJ}}_k)^2 = d_k^{\star}$.
\end{proof}

It is important to stress that \cref{prop:fj_below_baseline} concerns the \emph{guarantee}, not the realized performance of the algorithm: on any particular instance one may solve the SDP, round, and obtain a cut fraction that does in fact exceed the random baseline. The proposition says only that the worst-case certificate stops \emph{proving} anything beyond random once $d > d_k^{\star}$. For this reason we do not rely on the Frieze--Jerrum guarantee in what follows; our aim is instead an explicit algorithm that \emph{provably} beats the random baseline in exactly this regime.

Our explicit vector algorithm achieves cut fraction $(k-1)/k + \Theta(1/\sqrt{d})$, providing a provable guarantee that exceeds the random baseline on $d$-regular graphs with girth $g \geq 2m$. While the SDP provides strong theoretical guarantees, solving it requires $O(n^{3.5})$ time using interior-point methods, which becomes impractical for large graphs. Our algorithm requires only the minimum eigenvector of an $m \times m$ matrix, which is independent of $n$.

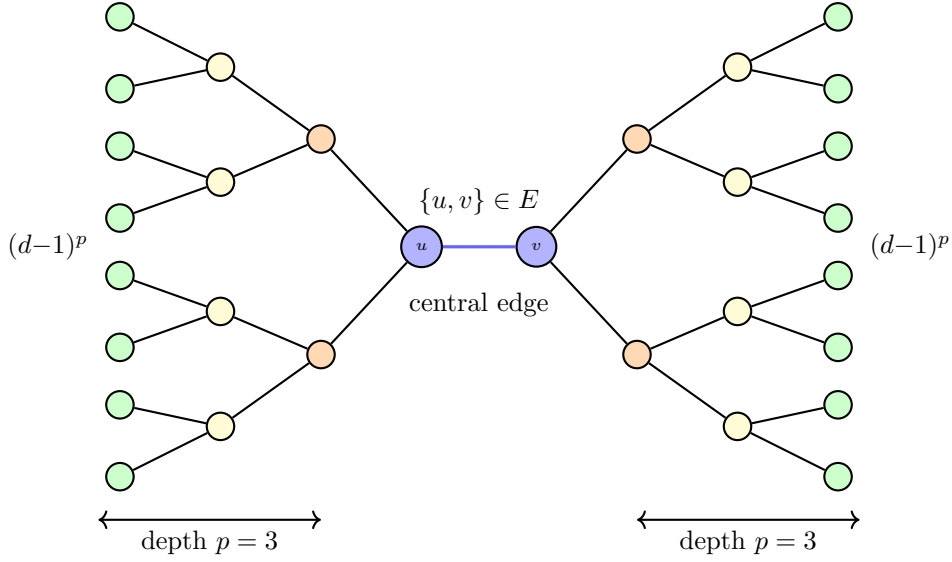
\begin{figure}[!htbp]
    \centering
    \begin{tikzpicture}[scale=0.95,
        vertex/.style={circle, draw, thick, minimum size=10pt, font=\tiny},
        central/.style={fill=blue!30},
        level1/.style={fill=orange!30},
        level2/.style={fill=yellow!20},
        level3/.style={fill=green!20},
        edge/.style={thick},
        centraledge/.style={very thick, blue!60},
    ]
    \node[vertex, central] (u) at (-0.8, 0) {$u$};
    \node[vertex, central] (v) at (0.8, 0) {$v$};
    \draw[centraledge] (u) -- (v);
    
    \node[vertex, level1] (u1) at (-2.2, 1.5) {};
    \node[vertex, level1] (u2) at (-2.2, -1.5) {};
    \draw[edge] (u) -- (u1);
    \draw[edge] (u) -- (u2);
    \node[vertex, level2] (u11) at (-3.6, 2.5) {};
    \node[vertex, level2] (u12) at (-3.6, 0.9) {};
    \node[vertex, level2] (u21) at (-3.6, -0.9) {};
    \node[vertex, level2] (u22) at (-3.6, -2.5) {};
    \draw[edge] (u1) -- (u11);
    \draw[edge] (u1) -- (u12);
    \draw[edge] (u2) -- (u21);
    \draw[edge] (u2) -- (u22);
    \node[vertex, level3] (u111) at (-5.0, 3.2) {};
    \node[vertex, level3] (u112) at (-5.0, 2.2) {};
    \node[vertex, level3] (u121) at (-5.0, 1.4) {};
    \node[vertex, level3] (u122) at (-5.0, 0.4) {};
    \node[vertex, level3] (u211) at (-5.0, -0.4) {};
    \node[vertex, level3] (u212) at (-5.0, -1.4) {};
    \node[vertex, level3] (u221) at (-5.0, -2.2) {};
    \node[vertex, level3] (u222) at (-5.0, -3.2) {};
    \draw[edge] (u11) -- (u111); \draw[edge] (u11) -- (u112);
    \draw[edge] (u12) -- (u121); \draw[edge] (u12) -- (u122);
    \draw[edge] (u21) -- (u211); \draw[edge] (u21) -- (u212);
    \draw[edge] (u22) -- (u221); \draw[edge] (u22) -- (u222);
    
    \node[vertex, level1] (v1) at (2.2, 1.5) {};
    \node[vertex, level1] (v2) at (2.2, -1.5) {};
    \draw[edge] (v) -- (v1);
    \draw[edge] (v) -- (v2);
    \node[vertex, level2] (v11) at (3.6, 2.5) {};
    \node[vertex, level2] (v12) at (3.6, 0.9) {};
    \node[vertex, level2] (v21) at (3.6, -0.9) {};
    \node[vertex, level2] (v22) at (3.6, -2.5) {};
    \draw[edge] (v1) -- (v11);
    \draw[edge] (v1) -- (v12);
    \draw[edge] (v2) -- (v21);
    \draw[edge] (v2) -- (v22);
    \node[vertex, level3] (v111) at (5.0, 3.2) {};
    \node[vertex, level3] (v112) at (5.0, 2.2) {};
    \node[vertex, level3] (v121) at (5.0, 1.4) {};
    \node[vertex, level3] (v122) at (5.0, 0.4) {};
    \node[vertex, level3] (v211) at (5.0, -0.4) {};
    \node[vertex, level3] (v212) at (5.0, -1.4) {};
    \node[vertex, level3] (v221) at (5.0, -2.2) {};
    \node[vertex, level3] (v222) at (5.0, -3.2) {};
    \draw[edge] (v11) -- (v111); \draw[edge] (v11) -- (v112);
    \draw[edge] (v12) -- (v121); \draw[edge] (v12) -- (v122);
    \draw[edge] (v21) -- (v211); \draw[edge] (v21) -- (v212);
    \draw[edge] (v22) -- (v221); \draw[edge] (v22) -- (v222);
    
    \node[below, font=\small] at (0.0, -0.5) {central edge};
    \draw[<->, thick] (-5.3, -3.8) -- (-2.2, -3.8) node[midway, below, font=\small] {depth $p=3$};
    \draw[<->, thick] (2.2, -3.8) -- (5.3, -3.8) node[midway, below, font=\small] {depth $p=3$};
    
    \node[above, font=\small] at (0, 0.3) {$\{u,v\} \in E$};
    \node[left, font=\small] at (-5.3, 0) {$(d{-}1)^p$};
    \node[right, font=\small] at (5.3, 0) {$(d{-}1)^p$};
    \end{tikzpicture}
    \caption{Edge-centered tree structure for a 3-regular graph ($d=3$) at depth $p=3$. The edge $\{u,v\}$ is at the center. Each endpoint branches into $d-1=2$ neighbors (excluding the central edge), and this branching continues for $p$ levels. On a graph of girth $\geq 2p+2 = 8$, this neighborhood contains no cycles, enabling exact computation of edge expectations via tree tensor network contraction. The same structure governs both QAOA (where correlations are confined to this tree) and the vector algorithm (where vector coefficients depend only on distance from each vertex).}
    \label{fig:tree_structure}
\end{figure}

\section{Local Vector Algorithm for Max-$k$-Cut}
\label{sec:vector_algorithm}

Our algorithm keeps the two ideas we borrowed from Frieze--Jerrum in \cref{subsec:sdp} (\emph{(i)} assign a vector to every vertex, and \emph{(ii)} round the vectors against $k$ random directions to obtain a label) but replaces the global SDP by an explicit, purely local construction. The starting point is the observation that on a $d$-regular graph of girth $g \geq 2m$, the $(m - 1)$-neighborhood of every edge is the \emph{same} tree (\cref{fig:tree_structure}). 
By the symmetry of this local picture, it is natural to assign each vertex a vector whose entries depend only on shortest-path distance, with entries corresponding to equally distant vertices taking the same value. The only freedom left is the choice of these values. Because every edge sees the identical local tree, the inner product between the vectors at the two endpoints of an edge is a single number $\sigma$ common to all edges. Choosing the coefficients to make adjacent vectors as anti-correlated as possible (i.e.\ minimizing $\sigma$) reduces to a small eigenvalue problem whose size is set by the girth, entirely independent of $n$. We then round the resulting vectors, first with the plain generalized Frieze--Jerrum rule and then with a strictly better neighbor-message rule.

\subsection{Vector Assignment}
\label{subsec:vector_assignment}

We construct an explicit vector assignment on $d$-regular graphs of girth $g \geq 2m$.
For each vertex $i \in V$, we assign a unit vector $\mathbf{v}_i \in \mathbb{R}^n$ with components
\begin{equation}
    v_{ij} = \begin{cases}
        \alpha_0 & \text{if } i = j, \\
        \alpha_\ell & \text{if } \text{dist}(i,j) = \ell < m, \\
        0 & \text{otherwise},
    \end{cases}
\end{equation}
where $\text{dist}(i,j)$ is the shortest-path distance between vertices $i$ and $j$, and $\{\alpha_\ell\}_{\ell=0}^{m-1}$ are parameters to be optimized. The key insight is that when $g \geq 2m$, the entries of $\mathbf{v}_i$ are determined solely by distance from $i$, and the structure is tree-like within radius $m-1$.

To illustrate, consider the Petersen graph (\cref{fig:vector_assignment}), which is 3-regular with girth 5. We use $m = 2$, so the vector assignment uses only $\alpha_0$ (self) and $\alpha_1$ (neighbors). Vertex 0 has neighbors $\{1, 4, 5\}$, so
\begin{equation}
    \mathbf{v}_0 = (\alpha_0, \alpha_1, 0, 0, \alpha_1, \alpha_1, 0, 0, 0, 0).
\end{equation}
Similarly, vertex 1 has neighbors $\{0, 2, 6\}$, giving $\mathbf{v}_1 = (\alpha_1, \alpha_0, \alpha_1, 0, 0, 0, \alpha_1, 0, 0, 0)$. For the edge $\{0,1\}$, the inner product is
\begin{equation}
    \mathbf{v}_0 \cdot \mathbf{v}_1 = \alpha_0 \cdot \alpha_1 + \alpha_1 \cdot \alpha_0 = 2\alpha_0 \alpha_1,
\end{equation}
where the two terms arise from vertices 0 and 1 themselves; there are no common neighbors since the graph is triangle-free. The unit-norm constraint requires $\alpha_0^2 + 3\alpha_1^2 = 1$. To maximize the cut probability, we want to minimize $\sigma = 2\alpha_0\alpha_1$, which leads to the eigenvalue problem described below.

\begin{figure}[!htbp]
    \centering
    \begin{tikzpicture}[scale=0.8]
    
    \tikzstyle{vtx}=[circle, draw, thick, minimum size=20pt, font=\small, fill=white]
    \tikzstyle{v0}=[vtx, fill=blue!30, line width=1.5pt]
    \tikzstyle{v1}=[vtx, fill=red!30, line width=1.5pt]
    \tikzstyle{other}=[vtx, fill=gray!15]
    
    \foreach \i in {0,...,4} {
        \coordinate (O\i) at ({90+72*\i}:2.5);
    }
    \foreach \i in {0,...,4} {
        \coordinate (I\i) at ({90+72*\i}:1.1);
    }
    
    \draw[thick, gray!60] (O0) -- (O1) -- (O2) -- (O3) -- (O4) -- (O0);
    \draw[thick, gray!60] (O0) -- (I0);
    \draw[thick, gray!60] (O1) -- (I1);
    \draw[thick, gray!60] (O2) -- (I2);
    \draw[thick, gray!60] (O3) -- (I3);
    \draw[thick, gray!60] (O4) -- (I4);
    \draw[thick, gray!60] (I0) -- (I2) -- (I4) -- (I1) -- (I3) -- (I0);
    
    \draw[very thick, purple!70] (O0) -- (O1);
    
    \node[v0] (n0) at (O0) {$0$};
    \node[v1] (n1) at (O1) {$1$};
    \node[other] (n2) at (O2) {$2$};
    \node[other] (n3) at (O3) {$3$};
    \node[other] (n4) at (O4) {$4$};
    \node[other] (n5) at (I0) {$5$};
    \node[other] (n6) at (I1) {$6$};
    \node[other] (n7) at (I2) {$7$};
    \node[other] (n8) at (I3) {$8$};
    \node[other] (n9) at (I4) {$9$};
    
    \node[anchor=south, font=\footnotesize, blue!70!black, overlay] at ($(O0)+(0, 0.5)$) {
        $\mathbf{v}_0 = (\alpha_0, \alpha_1, 0, 0, \alpha_1, \alpha_1, 0, 0, 0, 0)$
    };
    
    \node[anchor=south east, font=\footnotesize, red!70!black, overlay] at ($(O1)+(162:0.6)$) {
        $\mathbf{v}_1 = (\alpha_1, \alpha_0, \alpha_1, 0, 0, 0, \alpha_1, 0, 0, 0)$
    };
    
    \node[below, align=center] at (0, -3.2) {
        \small Petersen graph ($d=3$, girth $5$)
    };
    
    \end{tikzpicture}
    
    \caption{Vector assignment on the Petersen graph, a 3-regular graph with girth 5. Each vertex $i$ gets a vector $\mathbf{v}_i \in \mathbb{R}^{10}$ with entry $\alpha_0$ on coordinate $i$ (self), entry $\alpha_1$ on coordinates corresponding to neighbors, and $0$ elsewhere. Vertex 0 (blue) has neighbors $\{1, 4, 5\}$; vertex 1 (red) has neighbors $\{0, 2, 6\}$. These vertices share edge $\{0,1\}$ but no common neighbors (since girth $\geq 4$), so $\mathbf{v}_0 \cdot \mathbf{v}_1 = 2\alpha_0\alpha_1$.}
    \label{fig:vector_assignment}
\end{figure}

Because all edges in a regular graph connect vertices at distance 1 from each other, the inner product $\mathbf{v}_u \cdot \mathbf{v}_v$ takes the same value for every edge $\{u,v\} \in E$. We denote this common value by $\sigma$, which captures the correlation between the vectors assigned to adjacent vertices.

The unit-norm constraint $\|\mathbf{v}_i\|^2 = 1$ and the objective of minimizing $\sigma = \mathbf{v}_u \cdot \mathbf{v}_v$ (to maximize cut probability) lead to an eigenvalue problem.
After rescaling $\beta_\ell = \alpha_\ell \sqrt{d(d-1)^{\ell-1}}$ for $\ell \geq 1$ and $\beta_0 = \alpha_0$, we obtain
\begin{equation}
    \sigma = \min_{\|\boldsymbol{\beta}\|_2 = 1} \boldsymbol{\beta}^T A_m \boldsymbol{\beta} = \lambda_{\min}(A_m),
\end{equation}
where $A_m$ is the $m \times m$ tridiagonal matrix
\begin{equation}
    A_m = \begin{bmatrix}
        0 & a & 0 & \cdots & 0 \\
        a & 0 & b & \cdots & 0 \\
        0 & b & 0 & \ddots & \vdots \\
        \vdots & \ddots & \ddots & \ddots & b \\
        0 & \cdots & 0 & b & 0
    \end{bmatrix}, \quad
    a = \frac{1}{\sqrt{d}}, \quad b = \frac{\sqrt{d-1}}{d}.
    \label{eq:matrix_A}
\end{equation}

The optimal coefficients $\{\alpha_\ell\}$ are recovered from the minimum eigenvector of $A_m$.
At large $d$, the matrix $A_m$ approaches a symmetric tridiagonal Toeplitz matrix with off-diagonal entries $b \approx 1/\sqrt{d}$, whose eigenvalues are known in closed form \cite{Nochese2013}:
\begin{equation}
    \lambda_{\min}(A_m) \approx -\frac{2\cos(\pi/(m+1))}{\sqrt{d}} = -\frac{c_m}{\sqrt{d}},
\end{equation}
where we define $c_m = 2\cos(\pi/(m+1))$. For the depths $p=1,2,3$ (corresponding to girth requirements $4, 6, 8$), we have $c_2 = 1$, $c_3 = \sqrt{2} \approx 1.414$, and $c_4 = \varphi \approx 1.618$ (the golden ratio). This large-degree behavior is analyzed in detail in \cref{subsec:vector_asymptotics}. By construction of the above assignment, for all vertices $i, j$, $\mathbf{v}_i \cdot \mathbf{v}_j$ only depends on $\text{dist}(i, j)$. We define $\rho_L(d, m)$ as the value of this dot product for vertices at a distance $L$. Specifically,
\begin{align}
    \rho_L(d, m) & := \sum_{0 \leq \ell \leq m - 1}2(d - 1)^{\ell}\alpha_{\ell}\alpha_{L + \ell}\label{eq:rho_l_m_definition}\\
    & = 2\sqrt{\frac{d - 1}{d}}(d - 1)^{-L/2}\beta_0\beta_L + 2\frac{d - 1}{d}(d - 1)^{-L/2}\sum_{1 \leq \ell \leq m - 1}\beta_{\ell}\beta_{\ell + L}\\
    & = \frac{r_{L, m}}{d^{L/2}} + \mathcal{O}\left(\frac{1}{d^{(L + 1)/2}}\right),
\end{align}
for some constants $r_{L, m}$ depending only on $L, m$ but not on $d$.

\subsection{Rounding}
\label{subsec:rounding}

With the distance-shell vectors in hand, we sample $k$ independent random Gaussian vectors $\mathbf r_1, \ldots, \mathbf r_k \sim \mathcal N(0, I_n)$ and read off the Gaussian scores $X_v(a) = \mathbf v_v \cdot \mathbf r_a$ for every vertex $v$ and label $a \in [k]$. The simplest rule is the root-only argmax,
\begin{align*}
    \chi_0(v) = \arg\max_{a \in [k]} X_v(a),
\end{align*}
which already gives the TPM construction generalized to $k$ labels and has cut probability that depends only on the edge correlation $\sigma = \mathbf v_u \cdot \mathbf v_v$. Indeed each pair $(X_u(a), X_v(a))$ is bivariate Gaussian with correlation $\sigma$ and the pairs are independent across $a$, so by label symmetry the same-label probability equals $k$ times the probability that label $1$ wins at both endpoints. Conditioning on the winning scores $(X_u(1), X_v(1)) = (x, y)$ and integrating out the other $k-1$ labels, we get the closed form
\begin{equation}
    P_{\text{cut}}(\sigma, k) = 1 - k \int_{\mathbb R^2} \Phi_\sigma(x, y)^{k-1}\, \phi_\sigma(x, y)\, dx\, dy,
    \label{eq:cut_probability}
\end{equation}
where $\phi_\sigma, \Phi_\sigma$ are the bivariate Gaussian density and CDF at correlation $\sigma$. At $\sigma = 0$ this gives the baseline $(k-1)/k$; at $k = 2$ the well-known $\arccos(\sigma)/\pi$. \Cref{fig:cut_probability} verifies \eqref{eq:cut_probability} against Monte Carlo computation of the integral for $k \in \{2, 3, 4\}$.

\begin{figure}[t]
    \centering
    \includegraphics[width=0.5\textwidth]{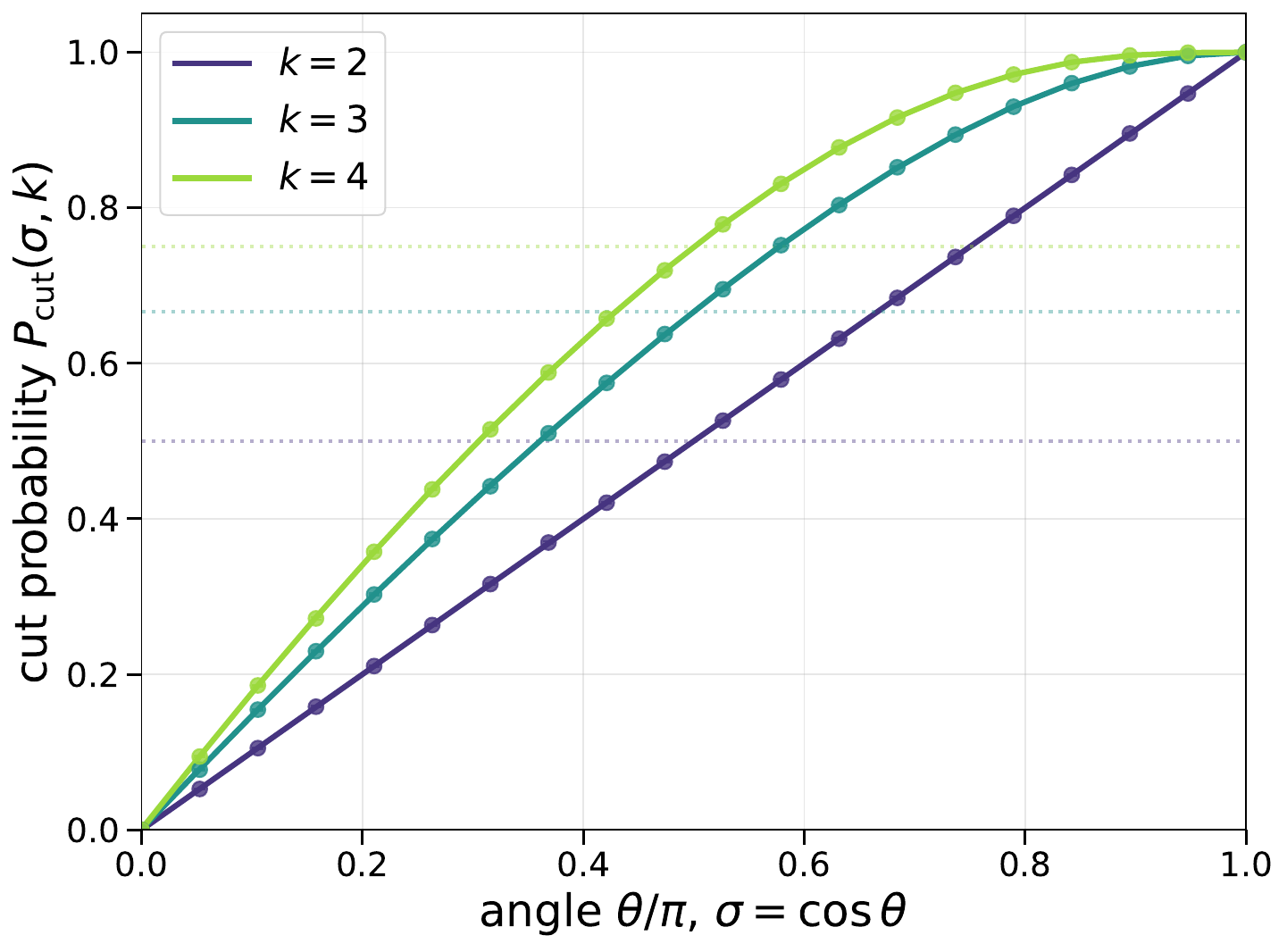}
    \caption{Cut probability $P_{\text{cut}}(\sigma, k)$ as a function of angle $\theta$ where $\sigma = \cos(\theta)$, for $k \in \{2, 3, 4\}$. Solid lines show the analytical formula \eqref{eq:cut_probability}; circles show Monte Carlo estimates ($100{,}000$ samples).}
    \label{fig:cut_probability}
\end{figure}

Root-only rounding throws away information that is already present in the Gaussian scores. Each neighbor $u$ of $v$ has, in addition to its favorite label $\arg\max_a X_u(a)$, a measure of how confident that preference is: the gap between its best and second-best scores. This gap is of the same order as the vector anti-correlation $\mathbf v_u \cdot \mathbf v_v \sim -c_m / \sqrt d$, so a small correction to the score at $v$ that listens to all $d$ neighbors can change the leading $1/\sqrt d$ behavior of the cut fraction. We make this precise by defining the gap message
\begin{equation}
    M_u(a) = \max_{b \neq a} X_u(b) - \max_b X_u(b),
    \label{eq:gap_message}
\end{equation}
which is non-positive and equal to zero unless $a$ is the favorite label of $u$, in which case it is the negative of the best-versus-second-best gap. The \emph{Local Vector} rule is
\begin{equation}
    \chi_\tau(v) = \arg\max_{a \in [k]} \biggl[\, X_v(a) + \frac{\tau}{\sqrt d}\sum_{u \in N(v)} M_u(a) \,\biggr],
    \label{eq:local_rule}
\end{equation}
with a single real parameter $\tau \geq 0$ controlling how strongly each vertex listens to its neighbors; $\tau = 0$ recovers the root-only rule.

The factor $1/\sqrt d$ in \eqref{eq:local_rule} is what makes the one-step correction sit at the leading order. After subtracting its label-symmetric mean, $\widetilde M_u(a) = M_u(a) - k^{-1}\sum_c M_u(c)$, the centered message has variance of order $1$. If the $d$ centered messages from $N(v)$ were independent their sum would have variance $O(d)$, so the prefactor $1/\sqrt d$ rescales the aggregate field to $O(1)$, the same scale as the original score $X_v(a)$. On the local tree the neighbor scores are not independent, but the distance-two correlation $\rho_2(d, m) = O(1/d)$ is small enough that the off-diagonal contribution to $\mathrm{Var}(\sum_u \widetilde M_u(a))$ is also $O(d)$. At the edge level the direct message $(\tau/\sqrt d) M_u(a)$ across $\{u, v\}$ has size $O(d^{-1/2})$, matching $\sigma = -c_m / \sqrt d + O(d^{-3/2})$ exactly. Both observations confirm that one round of neighbor messages can move the leading $1/\sqrt d$ coefficient rather than only a subleading correction.

The tree assumption required for the analysis fixes the girth requirement. The distance-shell construction needs the ball of radius $m - 1$ around every edge to be a tree, namely $g \geq 2m$. The Local Vector rule additionally needs the one-step edge neighborhood $\{o, o'\} \cup N(o) \cup N(o')$ to be a tree, which rules out the triangle and any $4$-cycle through the edge and gives the clean condition $g \geq 6$. 

\subsection{The Full Algorithm}
\label{subsec:algorithm}

Collecting the vector assignment and the rounding rule, the complete procedure is stated as \cref{alg:local_vector}. A single algorithm box treats both regimes uniformly: when $g \geq 6$ it applies the neighbor-message rule \eqref{eq:local_rule} with a tunable $\tau \geq 0$, and when $g < 6$ it falls back to root-only rounding ($\tau = 0$).

The expected edge cut probability $P^{(\tau)}_{d, m, k}$ achieved by \cref{alg:local_vector} depends only on $(d, m, k, \tau)$ under the girth condition $g \geq \max(2m, 6)$, because the relevant local tree has the same structure at every edge. At $\tau = 0$ the integral collapses to \eqref{eq:cut_probability} with $\sigma = \lambda_{\min}(A_m)$; for $\tau > 0$ it is a higher-dimensional Gaussian integral over the local tree. We evaluate and optimize this integral, and derive the resulting large-degree cut fraction, in \cref{subsec:vector_asymptotics}.

\begin{algorithm}[htb]
    \caption{Local Vector Algorithm for Max-$k$-Cut}
    \label{alg:local_vector}
    \KwIn{$d$-regular graph $G = (V, E)$ of girth $g$, number of disjoint sets $k$, depth parameter $m$ with $g \geq 2m$}
    \KwOut{$k$-cut $\chi : V \to [k]$}
    Compute $\sigma = \lambda_{\min}(A_m)$ and the minimum eigenvector $\boldsymbol\beta^\star$ of $A_m$ from \cref{eq:matrix_A}\;
    Recover coefficients $\alpha_\ell$ from $\boldsymbol\beta^\star$\;
    \For{each vertex $i \in V$}{
        Construct vector $\mathbf v_i$ with $v_{ij} = \alpha_{\text{dist}(i, j)}$ if $\text{dist}(i,j) < m$ and $v_{ij} = 0$ otherwise\;
    }
    Sample $k$ independent Gaussian vectors $\mathbf r_1, \ldots, \mathbf r_k \sim \mathcal N(0, I_n)$\;
    \For{each vertex $u \in V$ and label $a \in [k]$}{
        $X_u(a) \gets \mathbf v_u \cdot \mathbf r_a$\;
        $M_u(a) \gets \max_{b \neq a} X_u(b) - \max_b X_u(b)$\;
    }
    \eIf{$g \geq 6$}{
        Choose $\tau \geq 0$\;
        \For{each vertex $v \in V$}{
            $S_v(a) \gets X_v(a) + (\tau / \sqrt d) \sum_{u \in N(v)} M_u(a)$ for each $a \in [k]$\;
            $\chi(v) \gets \arg\max_{a \in [k]} S_v(a)$\;
        }
    }{
        \For{each vertex $v \in V$}{
            $\chi(v) \gets \arg\max_{a \in [k]} X_v(a)$\;
        }
    }
    \Return{$\chi$}
\end{algorithm}

\subsection{Analyzing the Cut Fraction}
\label{subsec:vector_asymptotics}

We now study how the expected cut fraction behaves as the degree $d$ grows with the girth parameter $m$ fixed. Throughout this subsection we work on $d$-regular graphs of girth $g \geq \max(2m, 6)$, so that the cut probability of \cref{alg:local_vector} is the single edge-tree integral $P^{(\tau)}_{d,m,k}$. Our goal is the leading correction to the random baseline, which for every choice of the rounding parameter $\tau \geq 0$ takes the form
\begin{equation}
    \mathbb{E}[\text{cut fraction}] = \frac{k-1}{k} + \frac{C_{m,k}(\tau)}{\sqrt{d}} + o\!\left(\frac{1}{\sqrt{d}}\right),
    \label{eq:asymptotic_form}
\end{equation}
so that the algorithm run at its optimal parameter achieves the coefficient
\begin{equation}
    C^{\star}_{m,k} = \sup_{\tau \geq 0} C_{m,k}(\tau).
    \label{eq:Cstar_def}
\end{equation}
When $g < 6$ the message rule does not apply and we set $\tau = 0$, so that $C^{\star}_{m,k} = C_{m,k}(0)$. We first compute the $\tau = 0$ coefficient in closed form, then explain the formula that computes the coefficient for $\tau > 0$. 

The analysis for $\tau = 0$ is driven entirely by the edge correlation $\sigma = \lambda_{\min}(A_m)$, whose large-degree expansion follows from a direct perturbation of the tridiagonal matrix in \eqref{eq:matrix_A}. Writing $A_m = M_m / \sqrt{d}$, the matrix $M_m$ has first off-diagonal entry $1$ and remaining off-diagonal entries $\sqrt{1 - 1/d} = 1 - \tfrac{1}{2d} + O(d^{-2})$. Let $T_m$ be the symmetric tridiagonal Toeplitz matrix with all off-diagonal entries equal to $1$, whose smallest eigenvalue is known in closed form to be $-c_m$ with
\begin{equation}
    c_m = 2\cos\!\left(\frac{\pi}{m+1}\right).
    \label{eq:cm_def}
\end{equation}
Since $\|M_m - T_m\|_{\mathrm{op}} = O(1/d)$, Weyl's inequality gives $\lambda_{\min}(M_m) = -c_m + O(1/d)$, and therefore
\begin{equation}
    \sigma = \lambda_{\min}(A_m) = -\frac{c_m}{\sqrt{d}} + O\!\left(\frac{1}{d^{3/2}}\right).
    \label{eq:sigma_asymp}
\end{equation}
The first few values are $c_2 = 1$, $c_3 = \sqrt{2}$, and $c_4 = \varphi = (1+\sqrt{5})/2$, and $c_m \uparrow 2$ as $m \to \infty$, reflecting that the achievable edge anti-correlation saturates at the spectral edge $-2/\sqrt{d}$ of the $d$-regular tree.

At $\tau = 0$ the cut probability is exactly $P_{\text{cut}}(\sigma, k)$ from \eqref{eq:cut_probability}, evaluated at $\sigma = \lambda_{\min}(A_m)$. Because $P_{\text{cut}}(\cdot, k)$ is smooth at $\sigma = 0$ with $P_{\text{cut}}(0, k) = (k-1)/k$, a first-order Taylor expansion gives
\begin{equation}
    P_{\text{cut}}(\sigma, k) = \frac{k-1}{k} - \alpha_k\, \sigma + O(\sigma^2),
    \qquad
    \alpha_k = -\left.\frac{\partial P_{\text{cut}}(\sigma, k)}{\partial \sigma}\right|_{\sigma = 0}.
    \label{eq:alpha_k_def}
\end{equation}
The constant $\alpha_k$ admits an explicit description in terms of the argmax regions
\begin{equation}
    R_a = \{\, y \in \mathbb{R}^k : y_a \geq y_b \ \text{for all } b \,\}, \qquad a \in [k].
    \label{eq:region_def}
\end{equation}
Writing $Z \sim \mathcal{N}(0, I_k)$ and $\mu^{(a)} = \mathbb{E}[\,Z\, \mathbf{1}\{Z \in R_a\}\,]$ for the truncated mean of $Z$ on the region $R_a$, one has
\begin{equation}
    \alpha_k = \sum_{a=1}^{k} \bigl\| \mu^{(a)} \bigr\|^2,
    \label{eq:alpha_k_integral}
\end{equation}
with $\alpha_2 = 1/\pi$ exactly; the values for larger $k$ are collected in \cref{tab:cm_alphak}. Substituting \eqref{eq:sigma_asymp} into \eqref{eq:alpha_k_def} yields the root-only coefficient
\begin{equation}
    C_{m,k}(0) = \alpha_k\, c_m.
    \label{eq:C_tau0}
\end{equation}
This is the generalization of the Thompson--Parekh--Marwaha guarantee to $k$ labels: since $\alpha_k > 0$ and $c_m > 0$, the expected cut fraction exceeds the random baseline $(k-1)/k$ by $\Theta(1/\sqrt{d})$ on every $d$-regular graph of girth $g \geq 2m$. The two ingredients of this coefficient are the girth factor $c_m$, which is set by the vector construction alone and grows towards its ceiling $2$ as the girth increases, and the rounding constant $\alpha_k$, which is set by the number of labels alone; their product $\alpha_k c_m$ is the coefficient of the root-only rule. \Cref{tab:cm_alphak} lists both.

\begin{table}[htb]
\centering
\begin{tabular}{c|ccccc}
\toprule
$m$ & 2 & 3 & 4 & 5 & 6 \\
\midrule
$c_m$ & 1.000 & 1.414 & 1.618 & 1.732 & 1.802 \\
\bottomrule
\end{tabular}
\quad
\begin{tabular}{c|ccccc}
\toprule
$k$ & 2 & 3 & 4 & 5 & 6 \\
\midrule
$\alpha_k$ & 0.318 & 0.358 & 0.353 & 0.338 & 0.321 \\
\bottomrule
\end{tabular}
\caption{The two factors of the root-only ($\tau = 0$) Local Vector coefficient $C_{m,k}(0) = \alpha_k c_m$. \emph{Left:} the girth factor $c_m = 2\cos(\pi/(m+1))$ from the eigenvalue expansion \eqref{eq:sigma_asymp} (so $c_2 = 1$, $c_3 = \sqrt2$, $c_4 = \varphi$), which increases towards $2$ as $m \to \infty$. \emph{Right:} the rounding constant $\alpha_k$ from the Gaussian integral \eqref{eq:alpha_k_integral}, with $\alpha_2 = 1/\pi$ exactly.}
\label{tab:cm_alphak}
\end{table}

For $\tau > 0$ the cut probability is an integral over the scores on the one-step edge neighborhood, and its leading coefficient is controlled by a central limit theorem for the external message sum $\sum_{u \in N(v)} M_u(a)$. The relevant statistics are those of a single vertex. For $z \in \mathbb{R}^k$ recall the gap message $m_a(z) = \max_{b \neq a} z_b - \max_b z_b$ from \eqref{eq:gap_message}, and let $\psi_a(z) = m_a(z) - \tfrac{1}{k}\sum_c m_c(z)$ be its label-centered version. Under $Z \sim \mathcal{N}(0, I_k)$ define the $k \times k$ matrices
\begin{equation}
    \Sigma_{ab} = \mathbb{E}[\psi_a(Z)\, \psi_b(Z)], \qquad
    Q_{pa} = \mathbb{E}[Z_p\, \psi_a(Z)], \qquad
    B = Q^{\top} Q.
    \label{eq:message_stats}
\end{equation}
The girth dependence enters only through the leading coefficients of the tree-distance vector overlaps, which we write as
\begin{equation}
    \rho_L(d, m) = \frac{r_{L,m}}{d^{L/2}} + O\!\left(\frac{1}{d^{(L+1)/2}}\right),
    \qquad r_{1,m} = -c_m,
    \label{eq:tree_constants}
\end{equation}
with $r_{2,m}$ and $r_{3,m}$ obtained in closed form from the minimum eigenvector of $A_m$. Assembling these into the matrices
\begin{align}
    \Lambda_\tau &= I - \tau c_m (Q + Q^{\top}) + \tau^2 (\Sigma + r_{2,m} B), \label{eq:Lambda_def} \\
    \Pi_\tau &= -c_m I + \tau r_{2,m} (Q + Q^{\top}) + \tau^2 r_{3,m} B, \label{eq:Pi_def}
\end{align}
we take $(X, Y)$ to be the jointly Gaussian pair in $\mathbb{R}^k \times \mathbb{R}^k$ with $\mathrm{Cov}(X) = I$, $\mathrm{Cov}(Y) = \Lambda_\tau$, and $\mathrm{Cov}(X, Y) = I - \tau c_m Q$, and for each label $a$ define the vectors
\begin{equation}
    \lambda^{(a)}(\tau) = \Lambda_\tau^{-1}\, \mathbb{E}\!\left[\,Y\, \mathbf{1}\{Y \in R_a\}\,\right],
    \qquad
    \nu^{(a)}(\tau) = \mathbb{E}\!\left[\,\psi(X)\, \mathbf{1}\{Y \in R_a\}\,\right].
    \label{eq:lambda_nu_def}
\end{equation}
\paragraph{Derivation of $C_{m,k}(\tau)$.}
Fix an edge $\{o,o'\}$ and write $\varepsilon=d^{-1/2}$. The score vectors $X_v=(X_v(1),\ldots,X_v(k))$ are jointly Gaussian, with
\begin{equation}
    \mathbb{E}\bigl[X_u X_w^{\top}\bigr] = (\mathbf v_u \cdot \mathbf v_w)\, I_k = \rho_{\mathrm{dist}(u,w)}(d,m)\, I_k,
\end{equation}
where \eqref{eq:tree_constants} gives $\rho_1=-c_m\varepsilon+o(\varepsilon)$, $\rho_2=r_{2,m}\varepsilon^2+o(\varepsilon^2)$, and $\rho_3=r_{3,m}\varepsilon^3+o(\varepsilon^3)$. Replacing $M_u$ by its label-centered version $\psi(X_u)$ does not change the $\arg\max$, because the removed term is the same for every label. The rounding score therefore splits as
\begin{equation}
    S_o=Y_o+\varepsilon\tau\psi(X_{o'}),\qquad
    Y_o=X_o+\tau F_o,\qquad
    F_o=\frac{1}{\sqrt d}\sum_{u\in N(o)\setminus\{o'\}}\psi(X_u),
    \label{eq:score_split}
\end{equation}
and symmetrically at $o'$. Here $F_o$ aggregates the $d-1$ messages entering $o$ from outside the focal edge, while the addend $\varepsilon\tau\psi(X_{o'})$ keeps the direct message across that edge separate.

Let $h_a(s)=\mathbf 1\{s\in R_a\}$. Label symmetry gives $\mathbb{E}h_a(S_v)=1/k$, and hence
\begin{equation}
    P^{(\tau)}_{d,m,k}
    =\frac{k-1}{k}-\sum_{a=1}^k
    \operatorname{Cov}\bigl(h_a(S_o),h_a(S_{o'})\bigr).
    \label{eq:cut_as_cov}
\end{equation}
Thus the gain over random is minus the same-label covariance across the edge.

For standard Gaussian vectors with $\mathbb{E}[XX'^{\top}]=\rho I_k$, Gaussian regression and Mehler's formula give
\begin{equation}
    \mathbb{E}[X\psi(X')^{\top}]=\rho Q,
    \qquad
    \operatorname{Cov}(\psi(X),\psi(X'))=\rho B+O(\rho^2).
    \label{eq:message_cov_expansion}
\end{equation}
Thus $\Sigma$ in \eqref{eq:message_stats} is the variance of one message, $Q$ is its linear correlation with the sender's score, and $B=Q^{\top}Q$ is the leading correlation between messages from weakly correlated senders.
In the tree neighborhood, two distinct neighbors of one endpoint are at distance $2$, whereas one outer neighbor of $o$ and one outer neighbor of $o'$ are at distance $3$. Applying \eqref{eq:message_cov_expansion} and counting these vertices gives
\begin{align}
    \mathbb{E}[X_oF_o^{\top}]&=-c_mQ+o(1), &
    \operatorname{Cov}(F_o)&=\Sigma+r_{2,m}B+o(1), \label{eq:field_moments}\\
    \operatorname{Cov}(Y_o)&=\Lambda_\tau+o(1), &
    \mathbb{E}[X_oY_o^{\top}]&=I-\tau c_mQ+o(1), \label{eq:effective_moments}\\
    \operatorname{Cov}(Y_o,Y_{o'})&=\varepsilon\Pi_\tau+o(\varepsilon). && \label{eq:effective_cross_cov}
\end{align}
For example, the $d-1$ individual variances in $F_o$ produce $\Sigma$, while its $\Theta(d^2)$ pairs of neighbors have correlation $\rho_2=\Theta(1/d)$ and produce $r_{2,m}B$. Across the edge, the bare scores produce $-c_mI$ in $\Pi_\tau$, the mixed score and field terms produce $\tau r_{2,m}(Q+Q^{\top})$, and the $\Theta(d^2)$ pairs of outer neighbors at distance $3$ produce $\tau^2r_{3,m}B$.

The central limit theorem enters precisely here. Each $F_o$ is a normalized sum of $d-1$ centered messages of bounded variance. Their pairwise correlations are $O(1/d)$, and a Hermite expansion shows that cumulants of order at least three vanish after the $d^{-1/2}$ normalization. Consequently $(X_o,Y_o)$ converges to the Gaussian pair $(X,Y)$ specified above \eqref{eq:lambda_nu_def}, with covariance \eqref{eq:effective_moments}; the two endpoints are coupled to first order by \eqref{eq:effective_cross_cov}. Thus the CLT applies to the collective field of the many neighbors of each endpoint, not to a single message.

It remains to translate these score correlations into label correlations. Price's theorem and Gaussian integration by parts give
\begin{equation}
    \mathbb{E}[\nabla h_a(Y)]
    =\Lambda_\tau^{-1}\mathbb{E}[Yh_a(Y)]
    =\lambda^{(a)}(\tau).
\end{equation}
Thus $\lambda^{(a)}$ is the sensitivity of the event that label $a$ wins. The covariance $\varepsilon\Pi_\tau$ between $Y_o$ and $Y_{o'}$ contributes $\varepsilon\lambda^{(a)\top}\Pi_\tau\lambda^{(a)}$. The direct message from $o'$ to $o$ contributes $\varepsilon\tau\lambda^{(a)\top}\nu^{(a)}$, because $\nu^{(a)}=\mathbb{E}[\psi(X)h_a(Y)]$ is the mean message sent when label $a$ wins; the reverse message contributes the same amount. Therefore
\begin{equation}
    \operatorname{Cov}\bigl(h_a(S_o),h_a(S_{o'})\bigr)
    =\varepsilon\Bigl[\lambda^{(a)\top}\Pi_\tau\lambda^{(a)}
    +2\tau\lambda^{(a)\top}\nu^{(a)}\Bigr]+o(\varepsilon).
\end{equation}
Substitution into \eqref{eq:cut_as_cov} gives
\begin{equation}
    C_{m,k}(\tau) = -\sum_{a=1}^{k} \Bigl[\, \lambda^{(a)}(\tau)^{\top}\, \Pi_\tau\, \lambda^{(a)}(\tau) + 2\tau\, \lambda^{(a)}(\tau)^{\top} \nu^{(a)}(\tau) \,\Bigr].
    \label{eq:working_formula}
\end{equation}
This coefficient is minus the leading same-label covariance across an edge. The matrix $\Pi_\tau$ records correlation transmitted through the many-neighbor fields, while $\nu^{(a)}$ records the direct instruction to avoid the neighbor's winning label. Every quantity in \eqref{eq:working_formula} is an expectation under a Gaussian of dimension at most $2k$, independent of $d$, so the coefficient can be evaluated by Monte Carlo integration. We have also verified \eqref{eq:working_formula} against direct Monte Carlo simulation of \cref{alg:local_vector} on the edge tree, and \cref{fig:loglog_validation} confirms the resulting $1/\sqrt d$ law and coefficients $C^{\star}_{m,k}$ numerically.

We can now state the guarantee for the full family of rounding rules.

\begin{theorem}
\label{thm:beats_random}
Let $k \geq 2$ and $m \geq 2$, and consider the Local Vector algorithm (\cref{alg:local_vector}) run with any fixed parameter $\tau \geq 0$ on a $d$-regular graph of girth $g \geq \max(2m, 6)$. Its expected cut fraction obeys
\begin{equation}
    \mathbb{E}[\emph{cut fraction}] = \frac{k-1}{k} + \frac{C_{m,k}(\tau)}{\sqrt{d}} + o\!\left(\frac{1}{\sqrt{d}}\right) \qquad (d \to \infty),
    \label{eq:thm_asymptotic}
\end{equation}
where $C_{m,k}(\tau)$ is given by \eqref{eq:working_formula}. It is continuous near $\tau=0$ and satisfies $C_{m,k}(0)=\alpha_kc_m>0$. Consequently, there is a threshold $\tau_0>0$ such that the expected cut fraction exceeds $(k-1)/k$ by $\Theta(1/\sqrt d)$ for every $\tau\in[0,\tau_0]$ and all sufficiently large $d$.
\end{theorem}

\begin{proof}
The central limit argument above establishes \eqref{eq:thm_asymptotic} and \eqref{eq:working_formula}. The Gaussian expectations in \eqref{eq:working_formula} vary continuously with $\tau$, as does $\Lambda_\tau^{-1}$ near $\tau=0$. Hence $C_{m,k}(\tau)$ is continuous near zero. At $\tau=0$, we have $\Lambda_0=I$, $\Pi_0=-c_mI$, and $\lambda^{(a)}(0)=\mu^{(a)}$. Therefore \eqref{eq:working_formula} reduces to $C_{m,k}(0)=c_m\sum_a\|\mu^{(a)}\|^2=\alpha_kc_m>0$. Continuity gives the claimed $\tau_0$, and \eqref{eq:thm_asymptotic} gives the stated improvement over the random baseline.
\end{proof}

At $\tau=0$, the exact finite-degree cut fraction is also available. The edge correlation is $\lambda_{\min}(A_m)$ at degree $d$, so
\begin{equation}
    \mathbb{E}[\text{cut fraction}] = P_{\text{cut}}\bigl(\lambda_{\min}(A_m),\, k\bigr).
    \label{eq:exact_cut_tau0}
\end{equation}
Inserting \eqref{eq:sigma_asymp} recovers \eqref{eq:C_tau0}. For $\tau>0$, the exact finite-degree value is the corresponding Gaussian integral over the one-step edge neighborhood with the exact overlaps $\rho_L(d,m)$. We evaluate this integral directly in \cref{subsec:finite_comparison}.

\begin{figure}[htb]
    \centering
    \includegraphics[width=\textwidth]{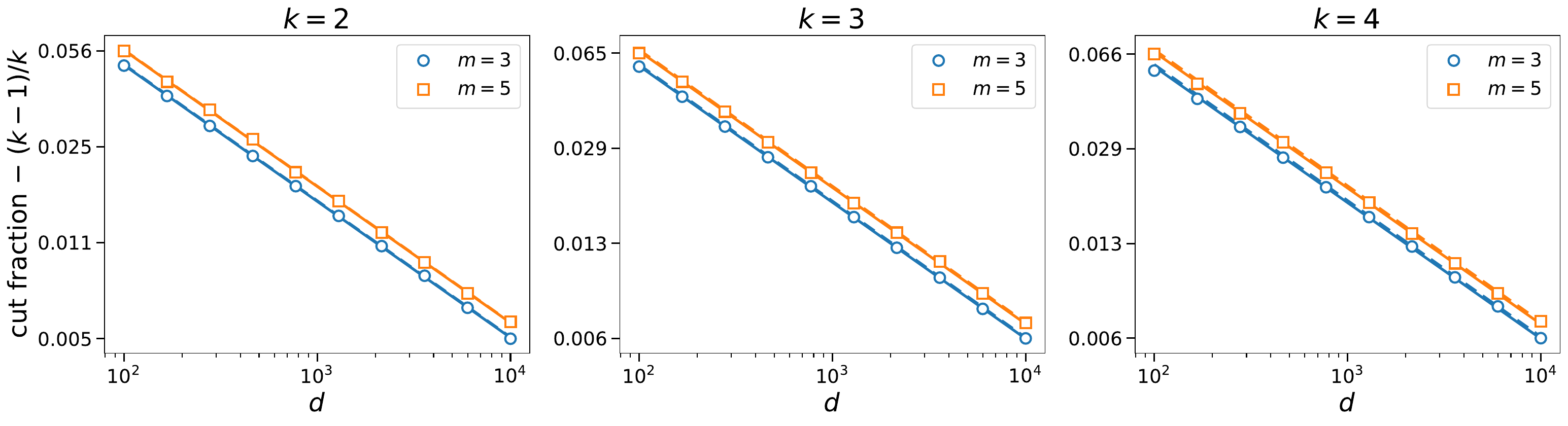}
    \caption{Numerical check of the large-degree asymptotics of \cref{thm:beats_random}. For each $k \in \{2,3,4\}$ and $m \in \{3,5\}$, we fix $\tau$ to the asymptotic optimum predicted by \eqref{eq:working_formula} and estimate the expected cut fraction on the local edge tree for degrees from $10^2$ to $10^4$. The measured gain over $(k-1)/k$ follows the $1/\sqrt d$ law, and the fitted coefficients agree with $C^{\star}_{m,k}$ to within about $0.01$.}
    \label{fig:loglog_validation}
\end{figure}

Optimizing \eqref{eq:working_formula} over $\tau$ gives the coefficients $C^{\star}_{m,k}$ reported in \cref{tab:asymptotic_coefficients}. For every $k \geq 2$ and every $m$ with $g \geq 6$ we find $C^{\star}_{m,k} > \alpha_k c_m$, so the neighbor-message rule strictly improves the leading $1/\sqrt{d}$ coefficient over the generalized Thompson--Parekh--Marwaha rounding, by up to roughly $0.10$ in absolute terms, with the largest gains at moderate $k$. The optimal parameter $\tau^{\star}$ is of order one and the coefficient is flat near its optimum over $\tau$. These coefficients are compared against the QAOA coefficient in \cref{subsec:infinite_comparison} (\cref{fig:infinite_comparison}).

However, the improvement is not uniform in the number of labels, and the case $k = 2$ behaves qualitatively differently from $k \geq 3$. For $k \geq 3$ the gap message is a genuinely nonlinear function of the neighbor scores, and the strict improvement $C^{\star}_{m,k} > \alpha_k c_m$ persists as the girth grows. For $k = 2$, by contrast, the message degenerates into a linear filter and the improvement, while strictly positive at every finite girth, washes out as the girth increases. 

Writing $D_v = X_v(1) - X_v(2)$ for the score difference, the gap message satisfies $M_v(1) - M_v(2) = -D_v$ identically, so the message rule reduces to the linear filter $\widetilde D_v = D_v - \tfrac{\tau}{\sqrt{d}} \sum_{u \in N(v)} D_u$. Because $(\widetilde D_o, \widetilde D_{o'})$ is jointly Gaussian, the cut probability is $\tfrac{1}{2} - \tfrac{1}{\pi}\arcsin \rho$ with $\rho$ the correlation across the edge, and $\rho$ is a weighted average of the one-step kernel over the spectrum of the $d$-regular tree, which is bounded below by the spectral edge $-2\sqrt{d-1}/d$. Consequently no choice of shell vectors or of $\tau$ can push the $k = 2$ coefficient above $2/\pi$, which the improved and root-only rules share as a common infinite-girth limit; this is the degeneracy anticipated above, absent for $k \geq 3$. The convergence of the two $k=2$ curves toward $2/\pi$ is visible in the left-most panel of \cref{fig:infinite_comparison}.

\begin{proposition}
\label{prop:k2_ceiling}
For $k = 2$, every girth $m$, and every $\tau \geq 0$, the Local Vector coefficient satisfies
\begin{equation}
    C_{m,2}(\tau) \leq \frac{2}{\pi},
\end{equation}
and the bound is approached only as $m \to \infty$ and $\tau \to 0$.
\end{proposition}

\section{Quantum Approximate Optimization Algorithm}
\label{sec:qaoa}

Having established what the best local classical algorithm can guarantee, we turn to the quantum side and analyze the Quantum Approximate Optimization Algorithm (QAOA) on the same class of graphs. Our aim is to place QAOA and the Local Vector algorithm on a common footing: we express the QAOA cut fraction, at fixed depth $p$ and on $d$-regular graphs of girth $g \geq 2p+2$, in the same $(k-1)/k + \Theta(1/\sqrt d)$ normalization used throughout \cref{sec:vector_algorithm}. As with the vector algorithm, the locally tree-like structure is what makes the analysis tractable, and the same tree neighborhoods of \cref{fig:tree_structure} govern both algorithms. We first set up the algorithm and reduce the cut fraction to an exact tensor contraction over the tree, then describe the finite-degree evaluation, and finally take the $d \to \infty$ limit through a qudit--boson reformulation that yields the leading coefficient.

\subsection{QAOA for Max-$k$-Cut}
\label{subsec:qaoa_framework}

The Quantum Approximate Optimization Algorithm (QAOA) \cite{farhi2014quantumapproximateoptimizationalgorithm} is a variational quantum algorithm for combinatorial optimization. For Max-$k$-Cut, each vertex is encoded as a qudit with $k$ levels, where the computational basis state $\ket{a}$ ($a \in \mathbb{Z}_k$) represents assigning label $a$ to that vertex.

The depth-$p$ QAOA prepares the state
\begin{equation}
    \ket{\bm{\gamma}, \bm{\beta}} = \left( \prod_{t=1}^{p} U_M(\beta_t) U_C(\gamma_t) \right) \ket{+}^{\otimes n},
\end{equation}
where $\ket{+} = k^{-1/2} \sum_{a=0}^{k-1} \ket{a}$ is the uniform superposition over labels. The \emph{phaser} $U_C(\gamma) = e^{-i\gamma H_C}$ encodes the cost function via the Hamiltonian
\begin{equation}
    H_C = \sum_{\{u,v\} \in E} P_{uv}, \qquad P_{uv} = \sum_{a \in \mathbb{Z}_k} \ket{aa}\bra{aa}_{uv},
\end{equation}
which penalizes edges with endpoints of the same label. The \emph{mixer} $U_M(\beta) = \bigotimes_{u \in V} U_{M,u}(\beta)$ generates transitions between labels; in this work we focus on the \emph{Grover mixer}
\begin{equation}
    U_{M}^{\text{Grover}}(\beta) = \bigotimes_{u \in V} e^{-i\beta \ket{+}\bra{+}_u}.
\end{equation}
The QAOA cost expectation is
\begin{equation}
    \langle C_G \rangle = \bra{\bm{\gamma}, \bm{\beta}} C_G \ket{\bm{\gamma}, \bm{\beta}} = |E| \cdot \bra{\bm{\gamma}, \bm{\beta}} (I - P_{uv}) \ket{\bm{\gamma}, \bm{\beta}},
\end{equation}
where the edge expectation is independent of the choice of edge on regular graphs, and it relates to the cut fraction via
\begin{equation}
    \text{cut fraction} = 1 - \frac{\langle H_C \rangle}{|E|} = 1 - \langle P_{uv} \rangle,
    \label{eq:qaoa_cut_fraction}
\end{equation}
with $\langle P_{uv} \rangle$ the probability that an edge has endpoints of the same label.

The locality of QAOA makes the edge expectation $\langle P_{uv} \rangle$ computable exactly when the graph has sufficiently large girth relative to the circuit depth $p$. Since the mixer acts vertex-by-vertex and the phaser acts on edges, the correlations produced by $p$ QAOA layers are confined to a subgraph of radius $p$ around each vertex. For an edge $\{u,v\}$, the expectation depends only on vertices within distance $p$ of $u$ and $v$, and if the girth satisfies $g \geq 2p + 2$ these neighborhoods form trees with no cycles (see \cref{fig:tree_structure}), so the edge expectation can be evaluated by contracting a tree tensor network \cite{qaoa_maxcut_high_depth, lower_bounding_max_cut_qaoa}. For example, depth $p=1$ requires girth $g \geq 4$ (triangle-free), $p=2$ requires $g \geq 6$, and $p=3$ requires $g \geq 8$, matching the girth requirements of the vector algorithm under the convention $m = p+1$. In particular, for a $d$-regular graph of girth $g \geq 2p+2$ the QAOA cut fraction depends only on $(k, d, p)$ and the chosen angles $(\bm{\gamma}, \bm{\beta})$, not on the specific graph instance.

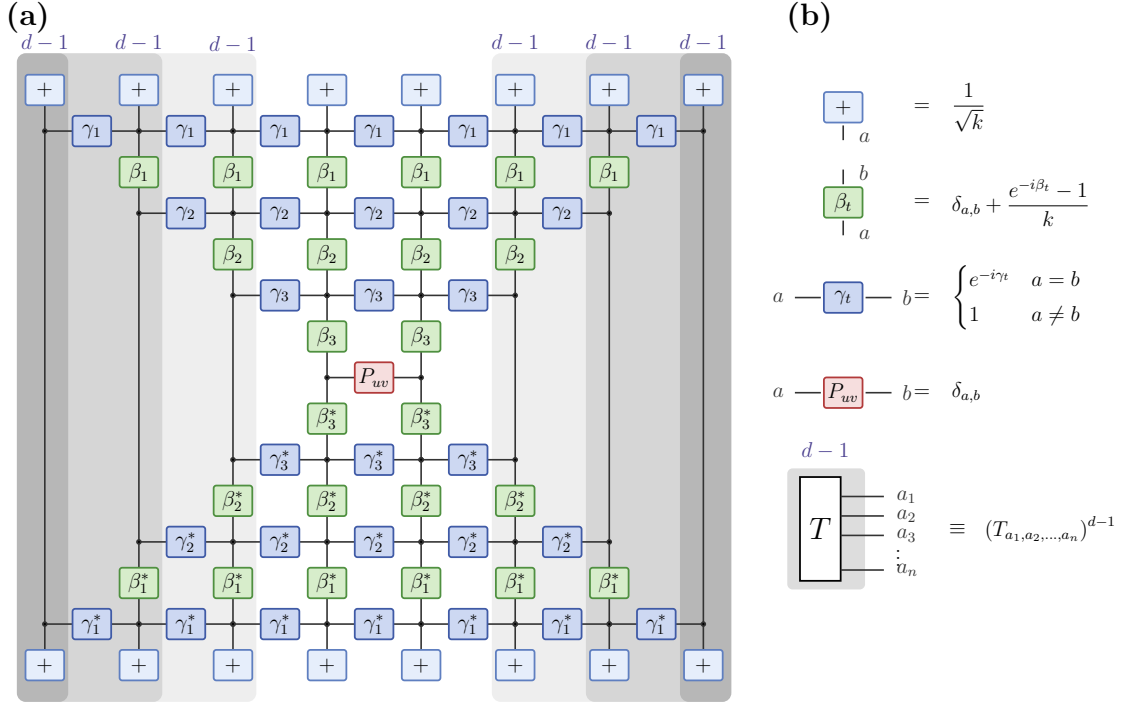
\begin{figure}[h!]
  \centering
  \definecolor{cplus}{HTML}{E6EEFB}
  \definecolor{cplusec}{HTML}{5F7FC0}
  \definecolor{cgamma}{HTML}{CDD8F2}
  \definecolor{cgammaec}{HTML}{3F5AA6}
  \definecolor{cbeta}{HTML}{D7EDCB}
  \definecolor{cbetaec}{HTML}{4F9040}
  \definecolor{cobs}{HTML}{F6DEDE}
  \definecolor{cobsec}{HTML}{B23B3B}
  \definecolor{cwire}{HTML}{333333}
  \definecolor{shade}{HTML}{8F7FCF}
  \definecolor{dlabc}{HTML}{4A3F8F}
  \definecolor{ctens}{HTML}{D9CFF2}
  \definecolor{cgray}{HTML}{555555}
  \tikzset{
    wire/.style={draw=cwire, line width=0.9pt},
    tbox/.style={rounded corners=1.5pt, line width=1.0pt, inner sep=0pt, font=\large, minimum width=0.760cm, minimum height=0.600cm},
    ilab/.style={text=cwire, font=\large},
    dlab/.style={text=dlabc, font=\large, anchor=south},
    tag/.style={anchor=south west, font=\LARGE},
    eqlab/.style={font=\large},
    vallab/.style={font=\large},
    tlab/.style={font=\LARGE},
  }
  \resizebox{0.9\textwidth}{!}{%
  \begin{tikzpicture}[x=1cm, y=0.85cm]
      \fill[cgray, opacity=0.10, rounded corners=5pt] (-0.55,-0.85) rectangle (4.16,14.15);
      \fill[cgray, opacity=0.10, rounded corners=5pt] (8.79,-0.85) rectangle (13.50,14.15);
      \node[dlab] at (3.7,14.00) {$d-1$};
      \node[dlab] at (9.25,14.05) {$d-1$};
      \fill[cgray, opacity=0.16, rounded corners=5pt] (-0.55,-0.85) rectangle (2.31,14.15);
      \fill[cgray, opacity=0.16, rounded corners=5pt] (10.64,-0.85) rectangle (13.50,14.15);
      \node[dlab] at (1.85,14.05) {$d-1$};
      \node[dlab] at (11.1,14.05) {$d-1$};
      \fill[cgray, opacity=0.24, rounded corners=5pt] (-0.55,-0.85) rectangle (0.46,14.15);
      \fill[cgray, opacity=0.24, rounded corners=5pt] (12.49,-0.85) rectangle (13.50,14.15);
      \node[dlab] at (0,14.05) {$d-1$};
      \node[dlab] at (12.95,14.05) {$d-1$};
      \draw[wire] (0,13.62) -- (0,-0.32);
      \draw[wire] (1.85,13.62) -- (1.85,-0.32);
      \draw[wire] (3.7,13.62) -- (3.7,-0.32);
      \draw[wire] (5.55,13.62) -- (5.55,-0.32);
      \draw[wire] (7.4,13.62) -- (7.4,-0.32);
      \draw[wire] (9.25,13.62) -- (9.25,-0.32);
      \draw[wire] (11.1,13.62) -- (11.1,-0.32);
      \draw[wire] (12.95,13.62) -- (12.95,-0.32);
      \node[tbox, fill=cplus, draw=cplusec] at (0,13.3) {$+$};
      \node[tbox, fill=cplus, draw=cplusec] at (0,0) {$+$};
      \node[tbox, fill=cplus, draw=cplusec] at (1.85,13.3) {$+$};
      \node[tbox, fill=cplus, draw=cplusec] at (1.85,0) {$+$};
      \node[tbox, fill=cplus, draw=cplusec] at (3.7,13.3) {$+$};
      \node[tbox, fill=cplus, draw=cplusec] at (3.7,0) {$+$};
      \node[tbox, fill=cplus, draw=cplusec] at (5.55,13.3) {$+$};
      \node[tbox, fill=cplus, draw=cplusec] at (5.55,0) {$+$};
      \node[tbox, fill=cplus, draw=cplusec] at (7.4,13.3) {$+$};
      \node[tbox, fill=cplus, draw=cplusec] at (7.4,0) {$+$};
      \node[tbox, fill=cplus, draw=cplusec] at (9.25,13.3) {$+$};
      \node[tbox, fill=cplus, draw=cplusec] at (9.25,0) {$+$};
      \node[tbox, fill=cplus, draw=cplusec] at (11.1,13.3) {$+$};
      \node[tbox, fill=cplus, draw=cplusec] at (11.1,0) {$+$};
      \node[tbox, fill=cplus, draw=cplusec] at (12.95,13.3) {$+$};
      \node[tbox, fill=cplus, draw=cplusec] at (12.95,0) {$+$};
      \draw[wire] (0,12.35) -- (1.85,12.35);
      \fill[wire] (0,12.35) circle (1.1pt);
      \fill[wire] (1.85,12.35) circle (1.1pt);
      \node[tbox, fill=cgamma, draw=cgammaec] at (0.925,12.35) {$\gamma_1$};
      \draw[wire] (1.85,12.35) -- (3.7,12.35);
      \fill[wire] (1.85,12.35) circle (1.1pt);
      \fill[wire] (3.7,12.35) circle (1.1pt);
      \node[tbox, fill=cgamma, draw=cgammaec] at (2.775,12.35) {$\gamma_1$};
      \draw[wire] (3.7,12.35) -- (5.55,12.35);
      \fill[wire] (3.7,12.35) circle (1.1pt);
      \fill[wire] (5.55,12.35) circle (1.1pt);
      \node[tbox, fill=cgamma, draw=cgammaec] at (4.625,12.35) {$\gamma_1$};
      \draw[wire] (5.55,12.35) -- (7.4,12.35);
      \fill[wire] (5.55,12.35) circle (1.1pt);
      \fill[wire] (7.4,12.35) circle (1.1pt);
      \node[tbox, fill=cgamma, draw=cgammaec] at (6.475,12.35) {$\gamma_1$};
      \draw[wire] (7.4,12.35) -- (9.25,12.35);
      \fill[wire] (7.4,12.35) circle (1.1pt);
      \fill[wire] (9.25,12.35) circle (1.1pt);
      \node[tbox, fill=cgamma, draw=cgammaec] at (8.325,12.35) {$\gamma_1$};
      \draw[wire] (9.25,12.35) -- (11.1,12.35);
      \fill[wire] (9.25,12.35) circle (1.1pt);
      \fill[wire] (11.1,12.35) circle (1.1pt);
      \node[tbox, fill=cgamma, draw=cgammaec] at (10.175,12.35) {$\gamma_1$};
      \draw[wire] (11.1,12.35) -- (12.95,12.35);
      \fill[wire] (11.1,12.35) circle (1.1pt);
      \fill[wire] (12.95,12.35) circle (1.1pt);
      \node[tbox, fill=cgamma, draw=cgammaec] at (12.025,12.35) {$\gamma_1$};
      \node[tbox, fill=cbeta, draw=cbetaec] at (1.85,11.4) {$\beta_1$};
      \node[tbox, fill=cbeta, draw=cbetaec] at (3.7,11.4) {$\beta_1$};
      \node[tbox, fill=cbeta, draw=cbetaec] at (5.55,11.4) {$\beta_1$};
      \node[tbox, fill=cbeta, draw=cbetaec] at (7.4,11.4) {$\beta_1$};
      \node[tbox, fill=cbeta, draw=cbetaec] at (9.25,11.4) {$\beta_1$};
      \node[tbox, fill=cbeta, draw=cbetaec] at (11.1,11.4) {$\beta_1$};
      \draw[wire] (1.85,10.45) -- (3.7,10.45);
      \fill[wire] (1.85,10.45) circle (1.1pt);
      \fill[wire] (3.7,10.45) circle (1.1pt);
      \node[tbox, fill=cgamma, draw=cgammaec] at (2.775,10.45) {$\gamma_2$};
      \draw[wire] (3.7,10.45) -- (5.55,10.45);
      \fill[wire] (3.7,10.45) circle (1.1pt);
      \fill[wire] (5.55,10.45) circle (1.1pt);
      \node[tbox, fill=cgamma, draw=cgammaec] at (4.625,10.45) {$\gamma_2$};
      \draw[wire] (5.55,10.45) -- (7.4,10.45);
      \fill[wire] (5.55,10.45) circle (1.1pt);
      \fill[wire] (7.4,10.45) circle (1.1pt);
      \node[tbox, fill=cgamma, draw=cgammaec] at (6.475,10.45) {$\gamma_2$};
      \draw[wire] (7.4,10.45) -- (9.25,10.45);
      \fill[wire] (7.4,10.45) circle (1.1pt);
      \fill[wire] (9.25,10.45) circle (1.1pt);
      \node[tbox, fill=cgamma, draw=cgammaec] at (8.325,10.45) {$\gamma_2$};
      \draw[wire] (9.25,10.45) -- (11.1,10.45);
      \fill[wire] (9.25,10.45) circle (1.1pt);
      \fill[wire] (11.1,10.45) circle (1.1pt);
      \node[tbox, fill=cgamma, draw=cgammaec] at (10.175,10.45) {$\gamma_2$};
      \node[tbox, fill=cbeta, draw=cbetaec] at (3.7,9.5) {$\beta_2$};
      \node[tbox, fill=cbeta, draw=cbetaec] at (5.55,9.5) {$\beta_2$};
      \node[tbox, fill=cbeta, draw=cbetaec] at (7.4,9.5) {$\beta_2$};
      \node[tbox, fill=cbeta, draw=cbetaec] at (9.25,9.5) {$\beta_2$};
      \draw[wire] (3.7,8.55) -- (5.55,8.55);
      \fill[wire] (3.7,8.55) circle (1.1pt);
      \fill[wire] (5.55,8.55) circle (1.1pt);
      \node[tbox, fill=cgamma, draw=cgammaec] at (4.625,8.55) {$\gamma_3$};
      \draw[wire] (5.55,8.55) -- (7.4,8.55);
      \fill[wire] (5.55,8.55) circle (1.1pt);
      \fill[wire] (7.4,8.55) circle (1.1pt);
      \node[tbox, fill=cgamma, draw=cgammaec] at (6.475,8.55) {$\gamma_3$};
      \draw[wire] (7.4,8.55) -- (9.25,8.55);
      \fill[wire] (7.4,8.55) circle (1.1pt);
      \fill[wire] (9.25,8.55) circle (1.1pt);
      \node[tbox, fill=cgamma, draw=cgammaec] at (8.325,8.55) {$\gamma_3$};
      \node[tbox, fill=cbeta, draw=cbetaec] at (5.55,7.6) {$\beta_3$};
      \node[tbox, fill=cbeta, draw=cbetaec] at (7.4,7.6) {$\beta_3$};
      \node[tbox, fill=cbeta, draw=cbetaec] at (1.85,1.9) {$\beta_1^{*}$};
      \node[tbox, fill=cbeta, draw=cbetaec] at (3.7,1.9) {$\beta_1^{*}$};
      \node[tbox, fill=cbeta, draw=cbetaec] at (5.55,1.9) {$\beta_1^{*}$};
      \node[tbox, fill=cbeta, draw=cbetaec] at (7.4,1.9) {$\beta_1^{*}$};
      \node[tbox, fill=cbeta, draw=cbetaec] at (9.25,1.9) {$\beta_1^{*}$};
      \node[tbox, fill=cbeta, draw=cbetaec] at (11.1,1.9) {$\beta_1^{*}$};
      \draw[wire] (0,0.95) -- (1.85,0.95);
      \fill[wire] (0,0.95) circle (1.1pt);
      \fill[wire] (1.85,0.95) circle (1.1pt);
      \node[tbox, fill=cgamma, draw=cgammaec] at (0.925,0.95) {$\gamma_1^{*}$};
      \draw[wire] (1.85,0.95) -- (3.7,0.95);
      \fill[wire] (1.85,0.95) circle (1.1pt);
      \fill[wire] (3.7,0.95) circle (1.1pt);
      \node[tbox, fill=cgamma, draw=cgammaec] at (2.775,0.95) {$\gamma_1^{*}$};
      \draw[wire] (3.7,0.95) -- (5.55,0.95);
      \fill[wire] (3.7,0.95) circle (1.1pt);
      \fill[wire] (5.55,0.95) circle (1.1pt);
      \node[tbox, fill=cgamma, draw=cgammaec] at (4.625,0.95) {$\gamma_1^{*}$};
      \draw[wire] (5.55,0.95) -- (7.4,0.95);
      \fill[wire] (5.55,0.95) circle (1.1pt);
      \fill[wire] (7.4,0.95) circle (1.1pt);
      \node[tbox, fill=cgamma, draw=cgammaec] at (6.475,0.95) {$\gamma_1^{*}$};
      \draw[wire] (7.4,0.95) -- (9.25,0.95);
      \fill[wire] (7.4,0.95) circle (1.1pt);
      \fill[wire] (9.25,0.95) circle (1.1pt);
      \node[tbox, fill=cgamma, draw=cgammaec] at (8.325,0.95) {$\gamma_1^{*}$};
      \draw[wire] (9.25,0.95) -- (11.1,0.95);
      \fill[wire] (9.25,0.95) circle (1.1pt);
      \fill[wire] (11.1,0.95) circle (1.1pt);
      \node[tbox, fill=cgamma, draw=cgammaec] at (10.175,0.95) {$\gamma_1^{*}$};
      \draw[wire] (11.1,0.95) -- (12.95,0.95);
      \fill[wire] (11.1,0.95) circle (1.1pt);
      \fill[wire] (12.95,0.95) circle (1.1pt);
      \node[tbox, fill=cgamma, draw=cgammaec] at (12.025,0.95) {$\gamma_1^{*}$};
      \node[tbox, fill=cbeta, draw=cbetaec] at (3.7,3.8) {$\beta_2^{*}$};
      \node[tbox, fill=cbeta, draw=cbetaec] at (5.55,3.8) {$\beta_2^{*}$};
      \node[tbox, fill=cbeta, draw=cbetaec] at (7.4,3.8) {$\beta_2^{*}$};
      \node[tbox, fill=cbeta, draw=cbetaec] at (9.25,3.8) {$\beta_2^{*}$};
      \draw[wire] (1.85,2.85) -- (3.7,2.85);
      \fill[wire] (1.85,2.85) circle (1.1pt);
      \fill[wire] (3.7,2.85) circle (1.1pt);
      \node[tbox, fill=cgamma, draw=cgammaec] at (2.775,2.85) {$\gamma_2^{*}$};
      \draw[wire] (3.7,2.85) -- (5.55,2.85);
      \fill[wire] (3.7,2.85) circle (1.1pt);
      \fill[wire] (5.55,2.85) circle (1.1pt);
      \node[tbox, fill=cgamma, draw=cgammaec] at (4.625,2.85) {$\gamma_2^{*}$};
      \draw[wire] (5.55,2.85) -- (7.4,2.85);
      \fill[wire] (5.55,2.85) circle (1.1pt);
      \fill[wire] (7.4,2.85) circle (1.1pt);
      \node[tbox, fill=cgamma, draw=cgammaec] at (6.475,2.85) {$\gamma_2^{*}$};
      \draw[wire] (7.4,2.85) -- (9.25,2.85);
      \fill[wire] (7.4,2.85) circle (1.1pt);
      \fill[wire] (9.25,2.85) circle (1.1pt);
      \node[tbox, fill=cgamma, draw=cgammaec] at (8.325,2.85) {$\gamma_2^{*}$};
      \draw[wire] (9.25,2.85) -- (11.1,2.85);
      \fill[wire] (9.25,2.85) circle (1.1pt);
      \fill[wire] (11.1,2.85) circle (1.1pt);
      \node[tbox, fill=cgamma, draw=cgammaec] at (10.175,2.85) {$\gamma_2^{*}$};
      \node[tbox, fill=cbeta, draw=cbetaec] at (5.55,5.7) {$\beta_3^{*}$};
      \node[tbox, fill=cbeta, draw=cbetaec] at (7.4,5.7) {$\beta_3^{*}$};
      \draw[wire] (3.7,4.75) -- (5.55,4.75);
      \fill[wire] (3.7,4.75) circle (1.1pt);
      \fill[wire] (5.55,4.75) circle (1.1pt);
      \node[tbox, fill=cgamma, draw=cgammaec] at (4.625,4.75) {$\gamma_3^{*}$};
      \draw[wire] (5.55,4.75) -- (7.4,4.75);
      \fill[wire] (5.55,4.75) circle (1.1pt);
      \fill[wire] (7.4,4.75) circle (1.1pt);
      \node[tbox, fill=cgamma, draw=cgammaec] at (6.475,4.75) {$\gamma_3^{*}$};
      \draw[wire] (7.4,4.75) -- (9.25,4.75);
      \fill[wire] (7.4,4.75) circle (1.1pt);
      \fill[wire] (9.25,4.75) circle (1.1pt);
      \node[tbox, fill=cgamma, draw=cgammaec] at (8.325,4.75) {$\gamma_3^{*}$};
      \draw[wire] (5.55,6.65) -- (7.4,6.65);
      \fill[wire] (5.55,6.65) circle (1.1pt);
      \fill[wire] (7.4,6.65) circle (1.1pt);
      \node[tbox, fill=cobs, draw=cobsec] at (6.475,6.65) {$P_{uv}$};
      \node[tag] at (-0.9,14.47) {\textbf{(a)}};
    \begin{scope}[xshift=14.350cm, yshift=1.105cm]
      \node[tag] at (0.1,13.17) {\textbf{(b)}};
      \node[tbox, fill=cplus, draw=cplusec] at (1.35,11.6) {$+$};
      \draw[wire] (1.35,11.18) -- (1.35,10.84);
      \node[ilab, anchor=west] at (1.53,10.9) {$a$};
      \node[eqlab] at (2.9,11.6) {$=$};
      \node[vallab, anchor=west] at (3.35,11.6) {$\dfrac{1}{\sqrt{k}}$};
      \node[tbox, fill=cbeta, draw=cbetaec] at (1.35,9.4) {$\beta_t$};
      \draw[wire] (1.35,9.82) -- (1.35,10.16);
      \draw[wire] (1.35,8.98) -- (1.35,8.64);
      \node[ilab, anchor=west] at (1.53,10.1) {$b$};
      \node[ilab, anchor=west] at (1.53,8.66) {$a$};
      \node[eqlab] at (2.9,9.4) {$=$};
      \node[vallab, anchor=west] at (3.35,9.4) {$\delta_{a,b}+\dfrac{e^{-i\beta_t}-1}{k}$};
      \draw[wire] (0.4,7.2) -- (0.91,7.2);
      \draw[wire] (1.79,7.2) -- (2.3,7.2);
      \node[tbox, fill=cgamma, draw=cgammaec] at (1.35,7.2) {$\gamma_t$};
      \node[ilab, anchor=east] at (0.32,7.2) {$a$};
      \node[ilab, anchor=west] at (2.38,7.2) {$b$};
      \node[eqlab] at (2.9,7.2) {$=$};
      \node[vallab, anchor=west] at (3.35,7.2) {$\begin{cases}e^{-i\gamma_t} & a=b\\[2pt]1 & a\neq b\end{cases}$};
      \draw[wire] (0.4,5) -- (0.91,5);
      \draw[wire] (1.79,5) -- (2.3,5);
      \node[tbox, fill=cobs, draw=cobsec] at (1.35,5) {$P_{uv}$};
      \node[ilab, anchor=east] at (0.32,5) {$a$};
      \node[ilab, anchor=west] at (2.38,5) {$b$};
      \node[eqlab] at (2.9,5) {$=$};
      \node[vallab, anchor=west] at (3.35,5) {$\delta_{a,b}$};
      \fill[cgray, opacity=0.2, rounded corners=3pt] (0.25,0.45) rectangle (1.78,3.25);
      \node[dlab] at (1.0,3.31) {$d-1$};
      \draw[fill=white, draw=black, line width=1pt] (0.5,0.65) rectangle (1.3,3.05);
      \node[tlab] at (0.9,1.85) {$T$};
      \draw[wire] (1.3,2.6) -- (2.15,2.6);
      \node[ilab, anchor=west] at (2.27,2.6) {$a_1$};
      \draw[wire] (1.3,2.15) -- (2.15,2.15);
      \node[ilab, anchor=west] at (2.27,2.15) {$a_2$};
      \draw[wire] (1.3,1.7) -- (2.15,1.7);
      \node[ilab, anchor=west] at (2.27,1.7) {$a_3$};
      \node at (2.43,1.3) {$\vdots$};
      \draw[wire] (1.3,0.9) -- (2.15,0.9);
      \node[ilab, anchor=west] at (2.27,0.9) {$a_n$};
      \node[eqlab] at (3.6,1.85) {$\equiv$};
      \node[vallab, anchor=west] at (4.0,1.85) {$\left(T_{a_1,a_2,\ldots,a_n}\right)^{d-1}$};
    \end{scope}
  \end{tikzpicture}%
  }%
  \caption{
  Exact tensor-network contraction of the depth-$p$ Max-$k$-Cut QAOA edge expectation $\langle P_{uv}\rangle$ on a $d$-regular graph of girth $g\geq 2p+2$, shown for $p=3$.
  Every wire carries a $k$-level qudit index $a\in\mathbb{Z}_k$.
  \textbf{(a)}
  The tensor network on the light cone of the focal edge $\{u,v\}$.
  From the $\ket{+}$ initial states at the top, each layer $t=1,\dots,p$ applies the cost phaser $\gamma_t$ on every edge followed by the Grover mixer $\beta_t$ on every qudit.
  The same-color projector $P_{uv}$ acts on the focal edge at the center, and the lower half is the conjugate mirror with $\gamma_t^{*},\beta_t^{*}$. Within a shaded region we contract all tensors together, then raise that tensor to the $(d-1)$th power as defined in panel~(b). Then do the same in the outer shaded regions, one by one. The whole contraction costs $\mathcal{O}(p\,k^{2(p+1)})$ time and $\mathcal{O}(k^{2p})$ memory, independent of $d$. 
  \textbf{(b)}
  The elementary tensors, given entry by entry: the uniform initial state $\ket{+}$, the Grover mixer $\beta_t$, the two-body
  cost phaser $\gamma_t=e^{-i\gamma_t P_{uv}}$, and the same-color projector $P_{uv}$.
  The bottom identity defines the entrywise (Hadamard) $(d-1)$th power that implements the tree symmetry: the branch tensor $T$ carrying the $d-1$ identical child subtrees equals $\bigl(T_{a_1,a_2,\dots,a_n}\bigr)^{d-1}$, i.e.\ every entry is raised to the power $d-1$.
  }
  \label{fig:finite_d_tn}
\end{figure}

\subsection{Exact Tensor Contraction Algorithm at Finite Degree}
\label{subsec:qaoa_tensor_algorithm}

To evaluate the edge expectation $\langle P_{uv} \rangle$ of
\eqref{eq:qaoa_cut_fraction} at fixed angles, we retain only the operator
$P_{uv}$ together with the QAOA gates contained in its light cone. Under the
girth condition $g \geq 2p+2$ this light cone is the tree of
\cref{fig:tree_structure}, and the expectation
$\langle \bm{\gamma}, \bm{\beta} | P_{uv} | \bm{\gamma}, \bm{\beta} \rangle$ is
computed by contracting a tensor network laid out on that tree, exactly as in
the qubit Max-Cut analyses of~\cite{lower_bounding_max_cut_qaoa}. Each phaser $U_C(\gamma)$ is diagonal in the computational
basis, so its two-body factor on an edge is a tensor over only two variables and
enters the network as a hyperedge; the mixers $U_M(\beta)$ and the initial
$\ket{+}$ states supply the remaining single-vertex tensors. The only change
from the qubit case is that we promote each vertex from a qubit to a $k$-level
qudit, so that every tensor index ranges over $\{0, 1, \ldots, k-1\}$ rather than
$\{0, 1\}$, and we update the gates accordingly: the observable is the
same-color projector $P_{uv}$ in place of $\tfrac{1}{2}(1 - Z_u Z_v)$, the phaser
is $U_C(\gamma) = e^{-i\gamma H_C}$ built from $P_{uv}$, and the single-qubit
mixer is replaced by the qudit Grover mixer $e^{-i\beta \ket{+}\bra{+}}$. The
same tree-contraction picture for qudit QAOA on regular graphs was used by~\cite{apte2026quantumapproximateoptimizationinteger}.

For a $d$-regular graph the number of qudits in the light cone of $P_{uv}$ grows
exponentially with $p$, as $2\,\bigl((d-1)^{p+1} - 1\bigr)/(d-2)$, so a direct
contraction is intractable already at moderate $p$. As observed in~\cite{lower_bounding_max_cut_qaoa}, 
however, every branch of the regular
tree is identical, so it suffices to contract a single branch inward toward the
root and raise its tensor entries entrywise to the $(d-1)$th power before moving
to the next level of the tree. This collapses all $d-1$ identical subtrees at
each vertex into one branch tensor and removes the exponential dependence on the
light-cone size. The branch tensor at depth $L$ carries $2L$ indices (one per
QAOA layer on each of the ket and bra sides) each ranging over the $k$ qudit
levels. The peak memory is therefore
$\mathcal{O}(k^{2p})$, the direct generalization of the
qubit cost of \cite{lower_bounding_max_cut_qaoa}.
\Cref{fig:finite_d_tn} illustrates the exact tensor network contraction of computing $\langle P_{uv} \rangle$
on a $d$-regular graph of girth $g\geq 2p+2$, shown for $p=3$. 
The running time is dominated by contracting the layer gates into the depth-$p$
branch tensor: each of the $p$ layers is a contraction of a $k^2 \times k^2$
gate (acting on the ket and bra index of that layer) across all $k^{2p}$
entries, at cost $\mathcal{O}(k^{2(p+1)})$, so building the branch costs
$\mathcal{O}(p\,k^{2(p+1)})$ in time, and the root contraction that joins the
two focal subtrees is of the same order. Both time and memory are independent of
the degree $d$.

\subsection{Infinite-Degree Analysis and the Qudit--Boson Map}
\label{subsec:qaoa_asymptotics}

To place QAOA in the same $(k-1)/k + \Theta(1/\sqrt d)$ form as the vector algorithm, we take the degree to infinity at fixed depth $p$ and fixed angles. It is convenient to work in a slightly more general setting than Max-$k$-Cut. Still assuming a $d$-regular graph $G = (V, E)$ of girth at least $2p + 2$, we consider a cost function of the form
\begin{align}
    C_d\left(\bm{x}\right) & = \frac{1}{\sqrt{d}}\sum_{\{u, v\} \in E}\varphi\left(x_u - x_v\right), && \bm{x} = \left(x_v\right)_{v \in V} \in \mathbb{Z}_k^V,\label{eq:general_cost_function}
\end{align}
where $\varphi: \mathbb{Z}_k \to \mathbb{R}$ is zero-mean and even,
\begin{align}
    \sum_{x \in \mathbb{Z}_k}\varphi\left(x\right) = 0, \qquad \varphi\left(-x\right) = \varphi\left(x\right) \quad \forall x \in \mathbb{Z}_k.
\end{align}
Any dit function can be centered to have zero mean, and evenness reflects the graph being undirected. For Max-$k$-Cut the relevant choice is $\varphi(c) = \mathbf{1}[c = 0] - 1/k$, and the $1/\sqrt d$ rescaling in \eqref{eq:general_cost_function} is exactly what makes the energy density have a nontrivial limit as $d \to \infty$ at fixed angles. Throughout we use the discrete Fourier transform of $\varphi$,
\begin{align}
    \hat{\varphi}(\xi) = \sum_{x \in \mathbb{Z}_k} \varphi(x)\, \overline{\omega}^{\,\xi x}, \qquad \varphi(c) = \frac{1}{k}\sum_{\xi \in \mathbb{Z}_k} \hat{\varphi}(\xi)\, \omega^{\xi c}, \qquad \omega := e^{2\pi i / k},
    \label{eq:phi_fourier}
\end{align}
noting that the zero-mean assumption is exactly $\hat{\varphi}(0) = 0$.

The starting point for the limit is a recursive form of the exact tree contraction, which we now state; it is one particular organization of the contraction of \cref{subsec:qaoa_tensor_algorithm}, convenient for the asymptotic analysis. Following Refs.~\cite{qaoa_maxcut_high_depth, apte2026quantumapproximateoptimizationinteger}, the edge expectation is assembled from a sequence of tensors $H_d^{(m)}(\bm{a})$, $0 \leq m \leq p$, indexed by a ditstring $\bm{a} \in \mathbb{Z}_k^{\mathcal{I}_p}$ whose entries are labeled by the QAOA layer indices
\begin{align}
    \mathcal{I}_p := \left\{1, 2, \ldots, p, p + 1, -p - 1, -p, \ldots, -2, -1\right\},
\end{align}
with positive indices from the bra and negative indices from the ket of $\bra{\bm{\gamma}, \bm{\beta}}P_{uv}\ket{\bm{\gamma}, \bm{\beta}}$. The tensors obey
\begin{align}
    H_d^{(0)}\left(\bm{a}\right) & := 1 && \forall \bm{a} \in \mathbb{Z}_k^{\mathcal{I}_p},\\
    H_d^{(m + 1)}\left(\bm{a}\right) & := \left(\sum_{\bm{b} \in \mathbb{Z}_k^{\mathcal{I}_p}}f\left(\bm{b}\right)H_d^{(m)}\left(\bm{b}\right)e^{i\Phi\left(\bm{a} - \bm{b}\right)/\sqrt{d}}\right)^d && 0 \leq m < p,\label{eq:h_iteration}
\end{align}
where
\begin{align}
    f\left(\bm{a}\right) & := \frac{1}{k}\prod_{1 \leq t \leq p}\bra{a_t}e^{i\beta_t\ket{+}\bra{+}}\ket{a_{t + 1}}\bra{a_{-t - 1}}e^{-i\beta_t\ket{+}\bra{+}}\ket{a_{-t}}
    \label{eq:f_tensor}
\end{align}
depends only on the mixer angles, and
\begin{align}
    \Phi\left(\bm{c}\right) & := \sum_{t \in \mathcal{I}_p}\Gamma_t\,\mathbf{1}\left[c_t = 0\right], \qquad \bm{\Gamma} = \left(\gamma_1, \ldots, \gamma_p, 0, 0, -\gamma_p, \ldots, -\gamma_1\right) \in \mathbb{R}^{\mathcal{I}_p},
    \label{eq:phi_phase}
\end{align}
the duplicated $\bm{\gamma}$ with opposite signs reflecting the bra and ket contributions. Given $H_d^{(p)}$, the per-edge energy is
\begin{align}
    \nu_d = \bra{\bm{\gamma}, \bm{\beta}}P_{uv}\ket{\bm{\gamma}, \bm{\beta}} - \frac{1}{k} & = \sum_{\bm{a}, \bm{b} \in \mathbb{Z}_k^{\mathcal{I}_p}}\varphi(a_{p+1} - b_{p+1})\,f\left(\bm{a}\right)H_d^{(p)}\left(\bm{a}\right)f\left(\bm{b}\right)H_d^{(p)}\left(\bm{b}\right)e^{i\Phi\left(\bm{a} - \bm{b}\right)/\sqrt{d}}.\label{eq:energy_from_h}
\end{align}

The analysis of the $d \to \infty$ limit rests on a translation-invariance property of these tensors under a global shift of all labels, which generalizes the $\mathbb{Z}_2$ symmetry of Max-Cut QAOA. It descends from the invariance of the Grover mixer and of the cost Hamiltonian \eqref{eq:general_cost_function} under the qudit shift $X = \sum_{x} \ket{x+1}\bra{x}$, i.e.\ $X^{\otimes n} U X^{\dagger \otimes n} = U$. At the level of the tensors this reads as follows.

\begin{lemma}
\label{lemma:f_h_translation_invariance}
For a translation-invariant mixer and the cost \eqref{eq:general_cost_function}, the tensors $f$ and $H_d^{(m)}$ are invariant under a uniform shift of the ditstring: for all $c \in \mathbb{Z}_k$, all $\bm{a} \in \mathbb{Z}_k^{\mathcal{I}_p}$, and all $0 \leq m \leq p$,
\begin{align}
    f\left(\bm{a} + c\bm{1}_{\mathcal{I}_p}\right) = f\left(\bm{a}\right), \qquad H_d^{(m)}\left(\bm{a} + c\bm{1}_{\mathcal{I}_p}\right) = H_d^{(m)}\left(\bm{a}\right).
\end{align}
\end{lemma}

\begin{proof}
For $f$, using $\ket{b + c} = X^c\ket{b}$ together with $X\ket{+} = \ket{+}$, so that $X^{\dagger c} e^{i\beta_t \ket{+}\bra{+}} X^c = e^{i\beta_t \ket{+}\bra{+}}$,
\begin{align}
    f\left(\bm{a} + c\bm{1}_{\mathcal{I}_p}\right) & = \frac{1}{k}\prod_{1 \leq t \leq p}\bra{a_t}X^{\dagger c}e^{i\beta_t \ket{+}\bra{+}}X^c\ket{a_{t + 1}}\bra{a_{-t - 1}}X^{\dagger c}e^{-i\beta_t \ket{+}\bra{+}}X^c\ket{a_{-t}} = f\left(\bm{a}\right).
\end{align}
For $H_d^{(m)}$ we induct on $m$; the base case $H_d^{(0)} = 1$ is immediate. Assuming the claim at level $m$, reindexing $\bm{b} \to \bm{b} + c\bm{1}_{\mathcal{I}_p}$ in \eqref{eq:h_iteration} and using the invariance of $f$ and $H_d^{(m)}$ together with $\Phi\big((\bm{a} + c\bm{1}_{\mathcal{I}_p}) - (\bm{b} + c\bm{1}_{\mathcal{I}_p})\big) = \Phi(\bm{a} - \bm{b})$,
\begin{align}
    H_d^{(m + 1)}\left(\bm{a} + c\bm{1}_{\mathcal{I}_p}\right) & = \left(\sum_{\bm{b}}f\left(\bm{b} + c\bm{1}_{\mathcal{I}_p}\right)H_d^{(m)}\left(\bm{b} + c\bm{1}_{\mathcal{I}_p}\right)e^{i\Phi(\bm{a} - \bm{b})/\sqrt{d}}\right)^d = H_d^{(m + 1)}\left(\bm{a}\right).
\end{align}
\end{proof}

With translation invariance in hand, the finite-degree recursion \eqref{eq:h_iteration} admits an explicit $d \to \infty$ limit.

\begin{proposition}[Infinite-degree limit of the tensors]
\label{prop:h_infinite_degree_limit}
For all fixed angles $\bm{\gamma}, \bm{\beta} \in \mathbb{R}^p$ and $0 \leq m \leq p$, the tensor $H_d^{(m)}$ converges entrywise to a $d$-independent limit $H^{(m)}$, given by the recursion $H^{(0)} \equiv 1$ and
\begin{align}
    H^{(m + 1)}\left(\bm{a}\right) = \exp\left(-\frac{1}{2}\sum_{\substack{\xi, \eta \in \mathbb{Z}_k\\t, u \in \mathcal{I}_p}}G^{(m)}_{(\xi, t),\,(\eta, u)}\,\Gamma_t\Gamma_u\,\frac{\hat{\varphi}\left(\xi\right)\hat{\varphi}\left(\eta\right)}{k^2}\,\omega^{\xi a_t}\omega^{\eta a_u}\right), \qquad 0 \leq m < p,
    \label{eq:H_limit_recursion}
\end{align}
where the correlation matrix $\bm{G}^{(m)}$ is defined from $H^{(m)}$ by
\begin{align}
    G^{(m)}_{(\xi, t),\,(\eta, u)} := \sum_{\bm{b} \in \mathbb{Z}_k^{\mathcal{I}_p}}\overline{\omega}^{\xi b_t}\overline{\omega}^{\eta b_u}\,f\left(\bm{b}\right)H^{(m)}\left(\bm{b}\right).
    \label{eq:G_def}
\end{align}
\end{proposition}

\begin{proof}
We induct on $m$. For $m = 0$, $H_d^{(0)}(\bm{a}) = 1$ for all $\bm{a}$, so the claim holds with $H^{(0)} \equiv 1$. Assume $H_d^{(m)}(\bm{a}) \to H^{(m)}(\bm{a})$ for all $\bm{a} \in \mathbb{Z}_k^{\mathcal{I}_p}$. Since $\mathcal{I}_p$ is finite the convergence is uniform and the entries are uniformly bounded in $d$.

Define $S_d^{(m)}(\bm{a}) := \sum_{\bm{b} \in \mathbb{Z}_k^{\mathcal{I}_p}} f(\bm{b}) H_d^{(m)}(\bm{b})\, e^{i\Phi(\bm{a} - \bm{b})/\sqrt{d}}$, so that $H_d^{(m+1)}(\bm{a}) = S_d^{(m)}(\bm{a})^d$. Expanding the exponential,
\begin{align}
    S_d^{(m)}(\bm{a}) & = \underbrace{\sum_{\bm{b}} f(\bm{b}) H_d^{(m)}(\bm{b})}_{=:\,S_d^{(m),\,0}(\bm{a})}
    + \underbrace{\frac{i}{\sqrt{d}}\sum_{\bm{b}} f(\bm{b}) H_d^{(m)}(\bm{b})\,\Phi(\bm{a} - \bm{b})}_{=:\,S_d^{(m),\,-1/2}(\bm{a})} \nonumber\\
    & \quad - \underbrace{\frac{1}{2d}\sum_{\bm{b}} f(\bm{b}) H_d^{(m)}(\bm{b})\,\Phi(\bm{a} - \bm{b})^2}_{=:\,S_d^{(m),\,-1}(\bm{a})} + O(d^{-3/2}),
\end{align}
where the remainder is $O(d^{-3/2})$ by uniform boundedness of $H_d^{(m)}$.

\paragraph{The order-$d^0$ term.}
By the tree-state normalization of QAOA on high-girth graphs \cite{qaoa_maxcut_high_depth, apte2026quantumapproximateoptimizationinteger}, $\sum_{\bm{b}} f(\bm{b}) H_d^{(m)}(\bm{b}) = 1$, so $S_d^{(m),\,0}(\bm{a}) = 1$.

\paragraph{The order-$d^{-1/2}$ term.}
Using $\Phi(\bm{a} - \bm{b}) = \sum_{t \in \mathcal{I}_p} \Gamma_t\, \mathbf{1}[a_t - b_t = 0]$,
\begin{align}
    S_d^{(m),\,-1/2}(\bm{a}) & = \frac{i}{\sqrt{d}}\sum_{t \in \mathcal{I}_p} i\Gamma_t \sum_{\bm{b} \in \mathbb{Z}_k^{\mathcal{I}_p}} f(\bm{b}) H_d^{(m)}(\bm{b})\,\varphi(a_t - b_t).
\end{align}
For each fixed $t$, introducing an average over the $k$ shifts $\bm{b} \to \bm{b} + c\bm{1}_{\mathcal{I}_p}$ (which does not change the sum) and applying \cref{lemma:f_h_translation_invariance},
\begin{align}
    \sum_{\bm{b}} f(\bm{b}) H_d^{(m)}(\bm{b})\,\varphi(a_t - b_t)
    & = \frac{1}{k}\sum_{c \in \mathbb{Z}_k} \sum_{\bm{b}} f(\bm{b}) H_d^{(m)}(\bm{b})\,\varphi(a_t - b_t - c) \nonumber\\
    & = \sum_{\bm{b}} f(\bm{b}) H_d^{(m)}(\bm{b}) \cdot \frac{1}{k}\sum_{c \in \mathbb{Z}_k}\varphi(a_t - b_t - c) = 0,
\end{align}
since $\sum_{c \in \mathbb{Z}_k}\varphi(c) = 0$. Thus $S_d^{(m),\,-1/2}(\bm{a}) = 0$.

\paragraph{The order-$d^{-1}$ term.}
Expanding $\Phi(\bm{a} - \bm{b})^2 = \sum_{t,u \in \mathcal{I}_p} \Gamma_t\Gamma_u\, \varphi(a_t - b_t)\varphi(a_u - b_u)$ and writing $\varphi$ via \eqref{eq:phi_fourier},
\begin{align}
    S_d^{(m),\,-1}(\bm{a}) & = -\frac{1}{2d}\sum_{t, u \in \mathcal{I}_p}\Gamma_t\Gamma_u \sum_{\xi, \eta \in \mathbb{Z}_k} \frac{\hat{\varphi}(\xi)\hat{\varphi}(\eta)}{k^2} \sum_{\bm{b}} f(\bm{b}) H_d^{(m)}(\bm{b})\,\omega^{\xi(a_t - b_t)}\omega^{\eta(a_u - b_u)}.
\end{align}
The inner sum factors as $\omega^{\xi a_t}\omega^{\eta a_u}$ times $\sum_{\bm{b}} f(\bm{b}) H_d^{(m)}(\bm{b})\,\overline{\omega}^{\xi b_t}\overline{\omega}^{\eta b_u}$. Passing $H_d^{(m)} \to H^{(m)}$ (induction hypothesis), the latter sum converges to $G^{(m)}_{(\xi, t),(\eta, u)}$, giving
\begin{align}
    S_d^{(m),\,-1}(\bm{a}) = -\frac{1}{2d}\sum_{\substack{\xi, \eta \in \mathbb{Z}_k \\ t, u \in \mathcal{I}_p}} G^{(m)}_{(\xi,t),(\eta,u)}\,\Gamma_t\Gamma_u\,\frac{\hat{\varphi}(\xi)\hat{\varphi}(\eta)}{k^2}\,\omega^{\xi a_t}\omega^{\eta a_u} + o_d(1/d).
\end{align}

\paragraph{Conclusion.}
Combining, $S_d^{(m)}(\bm{a}) = 1 - \frac{1}{2d}\sum_{\xi,\eta,t,u} G^{(m)}_{(\xi,t),(\eta,u)}\Gamma_t\Gamma_u \frac{\hat\varphi(\xi)\hat\varphi(\eta)}{k^2}\omega^{\xi a_t}\omega^{\eta a_u} + o_d(1/d)$. Raising to the $d$-th power and sending $d \to \infty$ via $(1 + c/d + o(1/d))^d \to e^c$ establishes \eqref{eq:H_limit_recursion}.
\end{proof}

The infinite-degree per-edge energy follows by the same expansion applied to \eqref{eq:energy_from_h}.

\begin{proposition}[Infinite-degree energy density]
\label{prop:energy_asymptotics}
In the $d \to \infty$ limit the per-edge energy \eqref{eq:energy_from_h} satisfies
\begin{align}
    \nu_d = \frac{\nu_{\infty}}{\sqrt{d}} + o\!\left(d^{-1/2}\right), \qquad \nu_{\infty} := \sum_{\substack{t \in \mathcal{I}_p\\\xi, \eta \in \mathbb{Z}_k}} i\Gamma_t\,\frac{\hat{\varphi}(\xi)\hat{\varphi}(\eta)}{k^2}\,G^{(p)}_{(\xi, p + 1),\,(\eta, t)}\,G^{(p)}_{(-\xi, p + 1),\,(-\eta, t)},
    \label{eq:nu_infinity}
\end{align}
with $\bm{G}^{(p)}$ the matrix of \cref{prop:h_infinite_degree_limit}.
\end{proposition}

\begin{proof}
The energy per edge at finite $d$ is exactly \eqref{eq:energy_from_h}. Expanding the phase $e^{i\Phi(\bm{a}-\bm{b})/\sqrt{d}} = 1 + \frac{i}{\sqrt{d}}\Phi(\bm{a}-\bm{b}) + O(1/d)$ gives $\nu_d = \nu_d^0 + \nu_d^{-1/2} + O(1/d)$ where
\begin{align}
    \nu_d^0 & = \sum_{\bm{a}, \bm{b} \in \mathbb{Z}_k^{\mathcal{I}_p}} f(\bm{a}) H_d^{(p)}(\bm{a})\, f(\bm{b}) H_d^{(p)}(\bm{b})\,\varphi(a_{p+1} - b_{p+1}), \label{eq:nu0} \\
    \nu_d^{-1/2} & = \frac{i}{\sqrt{d}} \sum_{\bm{a}, \bm{b}} f(\bm{a}) H_d^{(p)}(\bm{a})\, f(\bm{b}) H_d^{(p)}(\bm{b})\,\varphi(a_{p+1} - b_{p+1})\,\Phi(\bm{a} - \bm{b}). \label{eq:nu12}
\end{align}

\paragraph{The order-$d^0$ term.}
Expressing $\varphi(a_{p+1} - b_{p+1}) = \frac{1}{k}\sum_\xi \hat\varphi(\xi)\,\omega^{\xi(a_{p+1} - b_{p+1})}$ and separating the sums,
\begin{align}
    \nu_d^0 = \frac{1}{k}\sum_{\xi \in \mathbb{Z}_k} \hat\varphi(\xi)\left(\sum_{\bm{a}} f(\bm{a}) H_d^{(p)}(\bm{a})\,\omega^{\xi a_{p+1}}\right)\left(\sum_{\bm{b}} f(\bm{b}) H_d^{(p)}(\bm{b})\,\omega^{-\xi b_{p+1}}\right).
\end{align}
For $\xi \neq 0$, setting $S_\xi := \sum_{\bm{a}} f(\bm{a}) H_d^{(p)}(\bm{a})\,\omega^{\xi a_{p+1}}$ and reindexing $\bm{a} \to \bm{a} + \bm{1}_{\mathcal{I}_p}$ yields
\begin{align}
    S_\xi = \sum_{\bm{a}} f(\bm{a}) H_d^{(p)}(\bm{a})\,\omega^{\xi(a_{p+1}+1)} = \omega^\xi S_\xi,
\end{align}
where we used \cref{lemma:f_h_translation_invariance}. Since $\omega^\xi \neq 1$ for $\xi \neq 0$, $S_\xi = 0$. For $\xi = 0$, $\hat\varphi(0) = 0$ by the zero-mean assumption. Hence $\nu_d^0 = 0$.

\paragraph{The order-$d^{-1/2}$ term.}
Substituting $\Phi(\bm{a} - \bm{b}) = \sum_{t \in \mathcal{I}_p} \Gamma_t\,\varphi(a_t - b_t)$ and writing both $\varphi$ factors via \eqref{eq:phi_fourier},
\begin{align}
    \nu_d^{-1/2} = \frac{i}{\sqrt{d}} \sum_{t \in \mathcal{I}_p} \Gamma_t \sum_{\xi, \eta \in \mathbb{Z}_k} \frac{\hat\varphi(\xi)\hat\varphi(\eta)}{k^2}
    \left(\sum_{\bm{a}} f(\bm{a}) H_d^{(p)}(\bm{a})\,\omega^{\xi a_{p+1}}\omega^{\eta a_t}\right)
    \left(\sum_{\bm{b}} f(\bm{b}) H_d^{(p)}(\bm{b})\,\omega^{-\xi b_{p+1}}\omega^{-\eta b_t}\right).
\end{align}
By \cref{prop:h_infinite_degree_limit}, the inner sums converge:
\begin{align}
    \sum_{\bm{a}} f(\bm{a}) H_d^{(p)}(\bm{a})\,\omega^{\xi a_{p+1}}\omega^{\eta a_t} & = G^{(p)}_{(\xi,\,p+1),\,(\eta,\,t)} + o_d(1), \label{eq:G_sum_a}\\
    \sum_{\bm{b}} f(\bm{b}) H_d^{(p)}(\bm{b})\,\omega^{-\xi b_{p+1}}\omega^{-\eta b_t} & = G^{(p)}_{(-\xi,\,p+1),\,(-\eta,\,t)} + o_d(1). \label{eq:G_sum_b}
\end{align}
Hence $\nu_d^{-1/2} = \nu_\infty / \sqrt{d} + o_d(d^{-1/2})$ with $\nu_\infty$ as stated.
\end{proof}

Applied to Max-$k$-Cut with $\varphi(c) = \mathbf{1}[c=0] - 1/k$ and the Grover mixer, \cref{prop:energy_asymptotics} puts the QAOA cut fraction \eqref{eq:qaoa_cut_fraction} in the form
\begin{align}
    \text{cut fraction}_{\text{QAOA}} = \frac{k-1}{k} + \frac{C^{\text{QAOA}}_{k, p}(\bm{\gamma}, \bm{\beta})}{\sqrt d} + o\!\left(\frac{1}{\sqrt d}\right),
    \label{eq:qaoa_coefficient}
\end{align}
with $C^{\text{QAOA}}_{k,p}$ read off from $-\nu_{\infty}$, matching the $1/\sqrt d$ scaling of the vector coefficient $C_{m,k}$. Unlike the closed-form vector coefficient, $C^{\text{QAOA}}_{k,p}$ must be extracted from the limiting recursion \eqref{eq:H_limit_recursion}--\eqref{eq:nu_infinity} for each choice of angles, and the angles must themselves be optimized over.

Evaluating the limiting recursion directly costs $O(k^{2p+2})$ in time and memory, which becomes impractical at large $p$. As we describe next, the recursion can instead be recast as the dynamics of a single $k$-level qudit coupled to a finite collection of bosonic modes, generalizing the qubit spin--boson mapping used for Max-Cut \cite{2w94-rymn}. This reformulation represents the limiting state as a matrix product state whose cost grows only mildly with $p$, and it is the engine we use to evaluate $C^{\text{QAOA}}_{k,p}$ at the depths needed for the comparison in \cref{sec:comparison}.

\subsection{Qudit--Boson Representation and MPS Algorithm}
\label{subsec:qudit_boson}

Evaluating the limiting recursion \eqref{eq:H_limit_recursion} by direct iteration is impractical at large $p$. We now describe an equivalent reformulation that replaces this ditstring-indexed recursion with the dynamics of a single $k$-level qudit coupled to $N = (k-1)(p+1)$ harmonic-oscillator modes, represented as a matrix product state (MPS). This construction generalizes the qubit spin--boson mapping of Ref.~\cite{2w94-rymn} from $k=2$ to arbitrary $k$, and the qubit mapping is recovered as the single-channel special case $k = 2$.

Since $\varphi$ is a zero-mean even class function on $\mathbb{Z}_k$ with $\hat\varphi(0)=0$, the edge interaction couples the two endpoints only through the $k-1$ nontrivial clock operators
\begin{equation}
    Z_\xi = \sum_{a \in \mathbb{Z}_k} \omega^{\xi a}\, \ket{a}\bra{a}, \qquad \xi = 1, \ldots, k-1,
    \label{eq:clock_ops}
\end{equation}
which satisfy $Z_\xi^\dagger = Z_{k-\xi}$ and reduce to the Pauli $Z$ for $k=2$. The presence of $k-1$ Fourier channels is the central new feature of the qudit problem compared to the qubit case.

Before constructing the state we reduce the amount of data that has to be computed. The correlator $\bm{G}^{(p)}$ carries two contour indices, each ranging over the $2p+2$ slots of $\mathcal{I}_p$, but these indices are not independent: the bra and ket halves of the Schwinger--Keldysh contour are related by complex conjugation, and reflecting a slot $r$ across the turning point of the contour to $2p+3-r$ leaves the correlator either unchanged or conjugated. Concretely, for all $1 \leq r,s \leq p+1$ and all channels $\xi, \xi'$,
\begin{align}
    G^{(p)}_{(\xi,\,2p+3-r),\,(\xi',\,s)} &= G^{(p)}_{(\xi,\,r),\,(\xi',\,s)}, &
    G^{(p)}_{(\xi,\,r),\,(\xi',\,2p+3-s)} &= \bigl[G^{(p)}_{(\xi,\,r),\,(\xi',\,s)}\bigr]^{*}, &
    G^{(p)}_{(\xi,\,2p+3-r),\,(\xi',\,2p+3-s)} &= \bigl[G^{(p)}_{(\xi,\,r),\,(\xi',\,s)}\bigr]^{*}.
    \label{eq:contour_reflection}
\end{align}
Every entry of $\bm{G}^{(p)}$ is therefore determined by the block that keeps both indices on the bra side, $1 \leq r,s \leq p+1$. We collect this block into the \emph{Hermitian corner}
\begin{equation}
    G^{\mathrm{herm}}_{\xi,r,\xi',s} := G^{(p)}_{(\xi,\,r),\,(k-\xi',\,s)}, \qquad 1 \leq r,s \leq p+1,\quad 1 \leq \xi,\xi' \leq k-1,
    \label{eq:gherm_def}
\end{equation}
where the label reflection $\xi' \mapsto k-\xi'$ pairs each channel with its conjugate $Z_{\xi'}^\dagger = Z_{k-\xi'}$ so that the resulting $(k-1)(p+1) \times (k-1)(p+1)$ matrix is Hermitian. The $\mathbb{Z}_k$ symmetry of the cost makes $G^{\mathrm{herm}}$ positive semidefinite, so it admits a block upper-triangular Cholesky factorization
\begin{equation}
    G^{\mathrm{herm}} = L^\dagger L,
    \label{eq:Lfac}
\end{equation}
with $(k-1) \times (k-1)$ blocks indexed by the time slots $r, s \in \{1, \ldots, p+1\}$. Substituting the reflections \eqref{eq:contour_reflection} into the energy \eqref{eq:nu_infinity} and collecting the two halves of the contour gives the closed form
\begin{equation}
    \nu_\infty = \mathrm{Re}\sum_{\xi,\xi'=1}^{k-1}\frac{\hat\varphi(\xi)^{*}\,\hat\varphi(\xi')}{k^2}\sum_{r=1}^{2p+2} i\,\Gamma_r\,
    G^{(p)}_{(\xi,\,r),\,(\xi',\,p+1)}\;G^{(p)}_{(k-\xi,\,r),\,(k-\xi',\,p+1)},
    \label{eq:nu_qudit}
\end{equation}
in which only the $\xi, \xi' \neq 0$ channels appear because $\hat\varphi(0) = 0$; for $k=2$ this collapses to the single-channel autocorrelation formula of the binary case \cite{2w94-rymn}. The task is thus reduced to computing the Hermitian corner $G^{\mathrm{herm}}$, and the qudit--boson evolution below produces it one time slot at a time.

The qudit--boson state lives in
\begin{equation}
    \mathcal{H} = \mathbb{C}^k \otimes \mathcal{F}_1 \otimes \cdots \otimes \mathcal{F}_N, \qquad N = (k-1)(p+1),
    \label{eq:hilbert}
\end{equation}
with one Fock space $\mathcal{F}_m$ per mode $(\xi, \ell)$, $\xi \in [k-1]$, $\ell \in \{1, \ldots, p+1\}$. The $\ell = p+1$ slot is an auxiliary register kept in the vacuum; it allows the final block column of the correlator $\bm{G}^{(p)}$ to be read off at the end. The construction is initialized at
\begin{equation}
    \ket{\Psi_0} = \ket{+} \otimes \ket{\overline{0}_N},
    \label{eq:psi0}
\end{equation}
with $\ket{+}$ the qudit uniform superposition and $\ket{\overline{0}_N}$ the bosonic vacuum on all $N$ modes. The correlator is seeded with $G^{\mathrm{herm}}_{\cdot,1,\cdot,1} = I_{k-1}$, the first time slot of the Hermitian corner \eqref{eq:gherm_def}, which is positive semidefinite and Cholesky-factorized as in \eqref{eq:Lfac}. Each layer $\ell = 1, \ldots, p$ consists of a qudit-controlled multimode displacement
\begin{equation}
    K(\bm{\alpha}^{(\ell)}) = \sum_{a=0}^{k-1} \ket{a}\bra{a} \otimes \bigotimes_{m=1}^{N} D\!\left(\alpha^{(\ell)}_{a,m}\right),
    \qquad D(\alpha) := \exp\!\big(\alpha\,\hat{c}^\dagger - \alpha^*\,\hat{c}\big),
    \label{eq:Kdef}
\end{equation}
followed by the single-qudit Grover mixer $U_{\mathrm{mix}}(\beta_\ell) = e^{-i\beta_\ell\ket{+}\bra{+}}$. The displacement amplitudes are determined by the running Cholesky factor $L$ via the inverse DFT
\begin{equation}
    \tilde\alpha^{(\ell)}_{m,\xi} = -\frac{i\gamma_\ell}{k}\,\hat\varphi(\xi)\,L_{m,(\xi,\ell)}, \qquad
    \alpha^{(\ell)}_{a,m} = \sum_{\xi=1}^{k-1} \omega^{a\xi}\,\tilde\alpha^{(\ell)}_{m,\xi},
    \label{eq:alpha_fourier}
\end{equation}
generalizing the $\pm\alpha$ spin-controlled displacement of the binary case. After applying layer $\ell$, the new block column of $G^{\mathrm{herm}}$ is recovered from the qudit-resolved annihilation expectations
\begin{equation}
    o_{m,a} = \bra{\Psi}\!\left(\ket{a}\bra{a} \otimes \hat{c}_m\right)\!\ket{\Psi}, \qquad o_{m,\xi} = \sum_a \omega^{-\xi a}\,o_{m,a},
    \label{eq:annihilation_expectations}
\end{equation}
and the new column $g$ solves the linear system $Wg = o$ with $W_{m,(\xi',t)} = -\frac{i}{k}\hat\varphi(\xi')\gamma_t L_{m,(\xi',t)}$, after which $G^{\mathrm{herm}}$ and $L$ are updated by a single block-Cholesky step. After all $p$ layers the full $G^{(p)}$ is reconstructed from $G^{\mathrm{herm}}$ via the Schwinger--Keldysh reflections, and $\nu_\infty$ is evaluated from \eqref{eq:nu_infinity}.

To make this computationally practical, the bosonic state is stored as a $k$-tuple of MPS
\begin{equation}
    \ket{\psi_a} = \bigl(\bra{a} \otimes I^{\otimes N}\bigr)\ket{\Psi}, \qquad a = 0, \ldots, k-1,
    \label{eq:ktuple}
\end{equation}
with each $\mathcal{F}_m$ truncated to Fock dimension $f$ (modes $\ket{\overline{0}},\ldots,\ket{\overline{f-1}}$) and the MPS bond dimension capped at $\chi$ via SVD truncation after each Grover-mixer step. The truncated displacement matrices are built at $O(f^2)$ cost via the Laguerre recurrence.

\begin{algorithm}[htb]
    \caption{$\texttt{AppLayer}$: apply layer $\ell$ to the $k$-tuple MPS and read off the $G^{\mathrm{herm}}$ column}
    \label{alg:apply_layer}
    \KwIn{$k$-tuple state $(\ket{\psi_0},\ldots,\ket{\psi_{k-1}})$; Fourier amplitudes $\tilde\alpha^{(\ell)}$; mixer angle $\beta_\ell$; current $L$}
    \KwOut{new $G^{\mathrm{herm}}$ column $g$; evolved state}
    \For{$a = 0$ \KwTo $k-1$}{
        $\alpha_{a,\cdot} \gets \sum_{\xi=1}^{k-1} \omega^{a\xi}\,\tilde\alpha^{(\ell)}_{\cdot,\xi}$\;
        $\ket{\psi_a} \gets \bigotimes_m D(\alpha_{a,m})\,\ket{\psi_a}$\;
    }
    Apply Grover mixer $U_{\mathrm{mix}}(\beta_\ell)$ across the $k$ components and SVD-truncate to bond dimension $\chi$\;
    $o_{m,a} \gets \bra{\psi_a}\hat{c}_m\ket{\psi_a}$ for $m = 1,\ldots,(k-1)\ell$\;
    $o_{m,\xi} \gets \sum_a \omega^{-\xi a}\,o_{m,a}$\;
    $g \gets$ solve $Wg = o$ with $W_{m,(\xi',t)} = -\tfrac{i}{k}\hat\varphi(\xi')\gamma_t L_{m,(\xi',t)}$\;
    \Return $g$, evolved state\;
\end{algorithm}

\begin{algorithm}[htb]
    \caption{Compute $\widetilde\nu_p(\bm\gamma, \bm\beta; f, \chi)$}
    \label{alg:compute_nu}
    \KwIn{angles $\bm\gamma, \bm\beta$; Fock dimension $f$; bond cap $\chi$}
    \KwOut{$\widetilde\nu_p \approx \nu_\infty$}
    Initialize $\ket{\Psi} \gets \ket{+}\ket{\overline{0}_N}$, $N = (k-1)(p+1)$\;
    Initialize $G^{\mathrm{herm}}, L \gets$ zero tensors of shape $(k-1, p+1, k-1, p+1)$\;
    $G^{\mathrm{herm}}_{\cdot,1,\cdot,1} \gets I_{k-1}$;\quad $L_{\cdot,1,\cdot,1} \gets I_{k-1}$\;
    \For{$\ell = 1$ \KwTo $p$}{
        $\tilde\alpha^{(\ell)} \gets$ from $L$ column $\ell$ via \eqref{eq:alpha_fourier}\;
        $(g, \ket{\Psi}) \gets \texttt{AppLayer}(\ket{\Psi}, \tilde\alpha^{(\ell)}, \beta_\ell, L)$\;
        Update $G^{\mathrm{herm}}$ with column $g$ and $g^\dagger$, diagonal block $I_{k-1}$\;
        $\bm z \gets$ solve $L^\dagger_{(\leq\ell)}\,\bm z = g$\;
        $L_{\cdot,\ell+1,\cdot,\ell+1} \gets (I_{k-1} - \bm z^\dagger\bm z)^{1/2}$\;
    }
    Reconstruct $G^{(p)}$ from $G^{\mathrm{herm}}$ via the contour reflections\;
    $\widetilde\nu_p \gets$ evaluate \eqref{eq:nu_infinity} from $G^{(p)}$\;
    \Return $\widetilde\nu_p$\;
\end{algorithm}

The algorithm is exact up to two independent truncation parameters: the Fock dimension $f$ (convergence is rapid once $f \gtrsim \max_{a,m}|\alpha_{a,m}|^2$, since the displaced states have $O(1)$ mean occupation) and the bond dimension $\chi$ (grows by at most a factor of $k$ per Grover-mixer layer before truncation). In practice both remain modest at depths well beyond the reach of the $O(k^{2p+2})$ brute-force recursion. For Max-$k$-Cut with the Grover mixer, the algorithm is applied at each candidate pair of angles $(\bm\gamma, \bm\beta)$ and the optimal angles are found by gradient-free optimization over the resulting $\widetilde\nu_p$, yielding the optimized coefficient $C^{\text{QAOA}}_{k,p} = \max_{\bm\gamma,\bm\beta}(-\widetilde\nu_p)$.

\section{Comparison of QAOA to the Local Vector Algorithm}
\label{sec:comparison}

We are now in a position to revisit the question that motivated this work. On the same family of $d$-regular, locally tree-like graphs, which algorithm cuts more edges? Both algorithms cut a $(k-1)/k$ fraction of edges for free, and both improve on this baseline only through correlations built up over a neighborhood whose radius is limited by the girth. For the Local Vector algorithm this neighborhood is the ball of radius $m-1 = p$ (with $m = p+1$) that fits inside the girth, whereas for QAOA it is the causal cone of a depth-$p$ circuit. Since the girth requirement is $g \geq 2p+2$ in both cases, we index the two algorithms by the same parameter $p$ and compare them at equal girth. Throughout, we hold fixed the generalized Thompson--Parekh--Marwaha (TPM) root-only rounding as a reference point, since it is the classical baseline algorithm that QAOA and the improved-rounding Local Vector algorithm are supposed to improve upon.

At finite degree $d$ the relevant quantity is the expected cut fraction, which we evaluate at each girth in \cref{subsec:finite_comparison}. As $d$ grows the cut fractions of all three methods collapse onto the random baseline, and the informative quantity becomes the leading coefficient of the $\Theta(1/\sqrt d)$ correction above it, which we study in the large-degree limit of \cref{subsec:infinite_comparison}. In both regimes the qualitative picture is the same. At the smallest girth $(p=1)$ the Local Vector algorithm is ahead of or tied with QAOA, but once the girth is large enough QAOA attains a higher cut fraction. 

\begin{figure}[htbp]
    \centering
    \includegraphics[width=\textwidth]{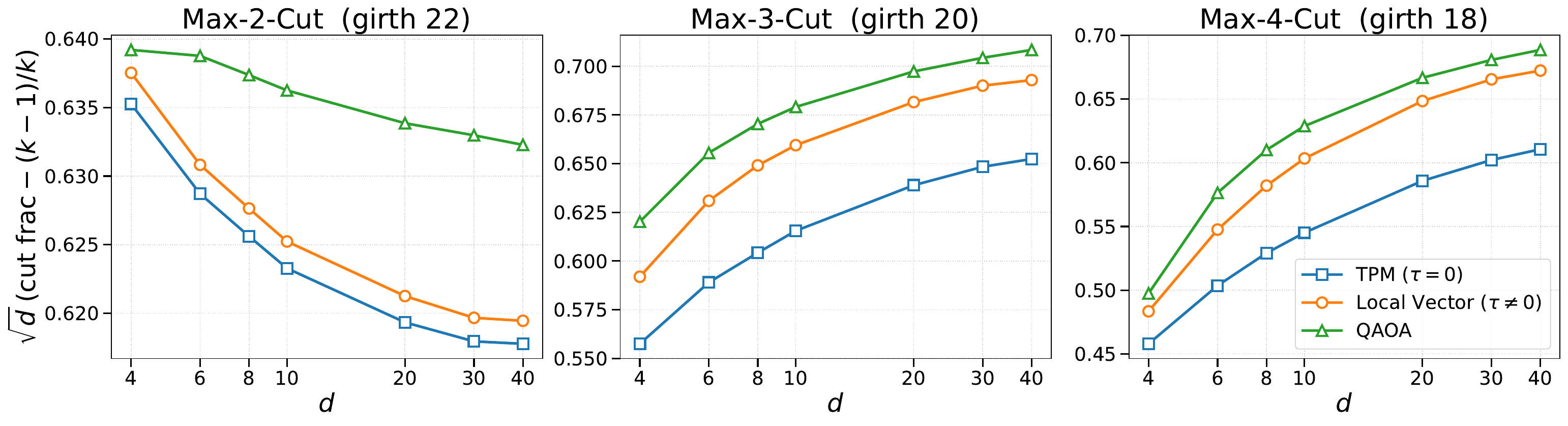}
    \caption{Expected cut fraction as a function of the degree $d$ on $d$-regular graphs of girth $\geq 2p+2$, shown separately for Max-2-Cut, Max-3-Cut, and Max-4-Cut, with each algorithm evaluated at the largest girth for which data is available. Each panel compares the improved-rounding Local Vector algorithm, the generalized Thompson--Parekh--Marwaha root-only rounding (TPM), and QAOA. The Local Vector and TPM values are the one-step tree averages of \cref{subsec:vector_asymptotics}, evaluated from the exact finite-degree expression \eqref{eq:cut_as_cov} (with the scores \eqref{eq:score_split} at the exact overlaps $\rho_L(d,m)$) and the root-only cut probability \eqref{eq:cut_probability} respectively, and the QAOA values are obtained by contracting the depth-$p$ tree tensor network of \cref{subsec:qaoa_tensor_algorithm} at optimized angles. Across all three orders the improved rounding keeps the Local Vector curve above TPM, and QAOA at the girths shown outperforms both. The tabulated values, together with the realized SDP cut fraction, are given in \cref{tab:comparison_k2,tab:comparison_k3,tab:comparison_k4}.}
    \label{fig:finite_comparison}
\end{figure}

\subsection{Finite-Degree Comparison}
\label{subsec:finite_comparison}

At finite degree we compare the three algorithms at the largest girth for which the QAOA tensor-network contraction remains within reach of our memory and runtime budget, so that every method is shown at its best. The Local Vector and TPM cut fractions are the one-step tree averages of \cref{subsec:vector_asymptotics}. For TPM this is the closed-form root-only value obtained by evaluating the bivariate Gaussian cut probability \eqref{eq:cut_probability} at the exact edge overlap $\sigma = \lambda_{\min}(A_m)$, with the integral evaluated by Monte Carlo, while for the Local Vector algorithm it is the $\tau$-optimized one-step neighborhood average \eqref{eq:cut_as_cov}, with the scores \eqref{eq:score_split} evaluated at the exact overlaps $\rho_L(d,m)$ rather than through the leading coefficient \eqref{eq:working_formula}. The QAOA cut fraction is obtained by contracting the depth-$p$ tree tensor network of \cref{subsec:qaoa_tensor_algorithm} at optimized angles $(\bm\gamma, \bm\beta)$.

\begin{figure}[htbp]
    \centering
    \includegraphics[width=\textwidth]{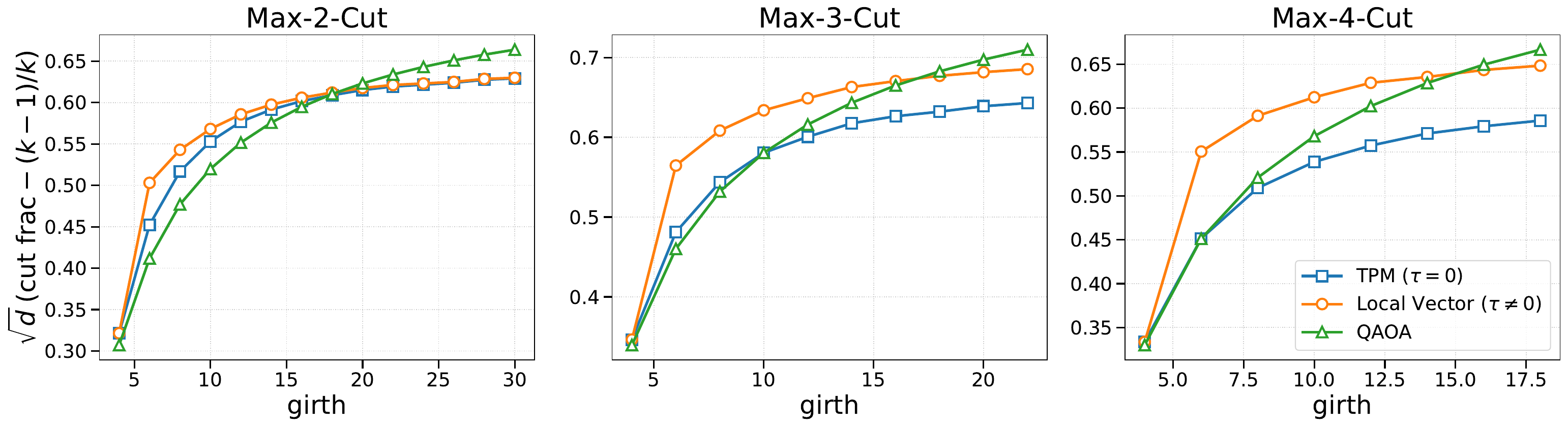}
    \caption{Expected cut fraction as a function of the girth (equivalently the depth $p$, with $m = p+1$) at the fixed degree $d = 20$, shown separately for Max-2-Cut, Max-3-Cut, and Max-4-Cut. Each panel compares the improved-rounding Local Vector algorithm, the generalized Thompson--Parekh--Marwaha root-only rounding (TPM), and QAOA; the values are computed as in \cref{fig:finite_comparison}, here at fixed $d = 20$ and varying girth. This is the finite-degree counterpart of the infinite-degree coefficient plot of \cref{fig:infinite_comparison}, with the same layout but the raw cut fraction on the vertical axis. All three methods improve with girth, and QAOA overtakes the Local Vector algorithm at moderate girth, while the improved rounding keeps the Local Vector curve above TPM.}
    \label{fig:finite_vs_girth}
\end{figure}

\Cref{fig:finite_comparison} shows the expected cut fraction as a function of the degree $d$ for each of Max-2-Cut, Max-3-Cut, and Max-4-Cut. The three panels tell a consistent story. The improved rounding lifts the Local Vector curve above the TPM curve for every $k \geq 2$, confirming that the neighbor-message rule leads to genuine improvement. QAOA, evaluated at the larger girths accessible to the tensor-network contraction, sits above the Local Vector curve across the range of degrees shown. The corresponding cut fraction numbers are collected in \cref{tab:comparison_k2,tab:comparison_k3,tab:comparison_k4} of \cref{app:tables}.

Holding the degree fixed instead isolates the dependence on the girth, which is the finite-degree analog of the infinite-degree crossover studied in \cref{subsec:infinite_comparison}. \Cref{fig:finite_vs_girth} plots the expected cut fraction against the girth (equivalently the depth $p$, with $m = p+1$) at the representative degree $d = 20$, with the same three methods and one panel per Max-$k$-Cut. The picture mirrors the infinite-degree coefficient plot of \cref{fig:infinite_comparison}, differing only in the vertical axis, which is the raw cut fraction here rather than its $1/\sqrt d$ coefficient. As the girth grows all three methods improve, but QAOA improves fastest, so its curve eventually crosses the Local Vector curve, while the improved rounding keeps the Local Vector curve above TPM for every $k \geq 2$. The values plotted are the $d = 20$ entries of \cref{tab:comparison_k2,tab:comparison_k3,tab:comparison_k4} across all available girths. 

\subsection{Infinite-Degree Comparison}
\label{subsec:infinite_comparison}

\begin{figure}[htbp]
    \centering
    \includegraphics[width=\textwidth]{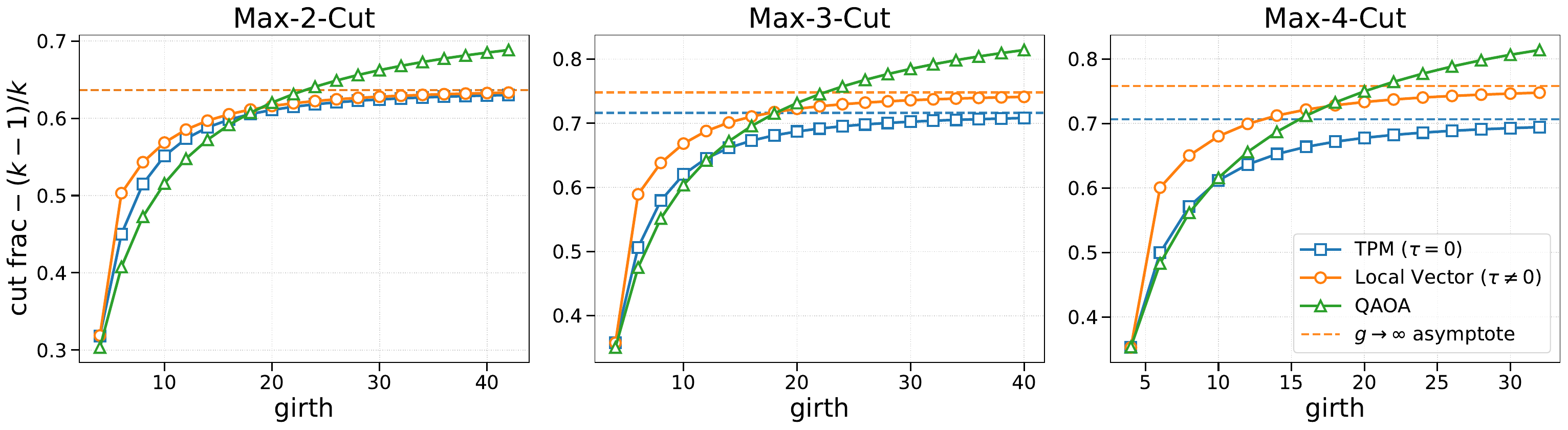}
    \caption{Leading coefficient of the $1/\sqrt d$ correction to the cut fraction above the random baseline $(k-1)/k$ in the infinite-degree limit, as a function of the girth (equivalently the depth $p$, with $m = p+1$), shown separately for Max-2-Cut, Max-3-Cut, and Max-4-Cut. Each panel compares the Local Vector coefficient $C^{\star}_{m,k} = \sup_{\tau \geq 0} C_{m,k}(\tau)$ of \cref{thm:beats_random}, its Thompson--Parekh--Marwaha root-only specialization $\alpha_k c_m$ at $\tau = 0$, and the optimized QAOA coefficient $C^{\text{QAOA}}_{k,p}$ of \eqref{eq:qaoa_coefficient}, computed from the qudit--boson recursion by \cref{alg:compute_nu}. The improved rounding opens a persistent gap between the Local Vector and TPM curves for every $k \geq 3$, and the QAOA curve overtakes the Local Vector curve by girth $g=20$ ($p=9$). The dashed horizontal lines mark the girth-to-infinity ($p \to \infty$) asymptotic coefficients of the Local Vector and TPM rules; for $k = 2$ the two coincide at the common ceiling $2/\pi$ (\cref{prop:k2_ceiling}), while for $k \geq 3$ they remain distinct. The tabulated values are given in \cref{tab:asymptotic_coefficients}.}
    \label{fig:infinite_comparison}
\end{figure}

In the large-degree limit,  each algorithm cuts a fraction
\begin{align}
    \frac{k-1}{k} + \frac{C_{k,p}}{\sqrt d} + o\!\left(\frac{1}{\sqrt d}\right)
\end{align}
of the edges, and the girth-dependent coefficient $C_{k,p}$ is the figure of merit. For the Local Vector algorithm this coefficient is $C^{\star}_{m,k} = \sup_{\tau \geq 0} C_{m,k}(\tau)$ with $m = p+1$, obtained by optimizing the asymptotic coefficient in \eqref{eq:working_formula} over the rounding strength, as established in \cref{thm:beats_random}, and the TPM reference is its $\tau = 0$ specialization $\alpha_k c_m$. The QAOA coefficient is $C^{\text{QAOA}}_{k,p}$ of \eqref{eq:qaoa_coefficient}, read off from the limiting value $-\nu_\infty$ of the qudit--boson recursion and extracted numerically at optimized angles by \cref{alg:compute_nu}.

\Cref{fig:infinite_comparison} plots these coefficients against the girth, again in one panel per Max-$k$-Cut. The improved rounding shows up here as a strict and persistent gap between the Local Vector and TPM curves for every $k \geq 3$ and every $p \geq 2$, and the gap does not close as the girth grows, so the improvement survives into the deep-circuit limit. The QAOA curve starts below or level with the Local Vector curve at $p = 1$ but climbs more steeply and eventually crosses it. Past the crossover the quantum algorithm has the larger $1/\sqrt d$ coefficient at equal girth, which is the precise sense in which we identify a regime of quantum advantage. The numerical values behind the plotted curves are tabulated in \cref{tab:asymptotic_coefficients} of \cref{app:tables}.

\section{Conclusion and Future Work}
\label{sec:conclusion}

In this paper we have developed both a state-of-the-art provable classical algorithm for Max-$k$-Cut, and the tools needed to analyze the performance of QAOA. On the classical side, the Local Vector algorithm assigns distance-shell vectors determined by a single small eigenvector computation and rounds them against $k$ random directions, and the one-step improved rounding strictly raises the leading $1/\sqrt d$ coefficient for every $k \geq 2$; because this improvement is additive and holds exactly where the Frieze--Jerrum guarantee is vacuous (\cref{prop:fj_below_baseline}), it gives the best provable cut-fraction guarantee we know of among efficient classical algorithms on regular graphs. On the quantum side, the qudit--boson reformulation opens up the infinite-degree limit of QAOA, and the degree-independent tree tensor-network contraction reaches the depths and degrees needed for a head-to-head comparison. Placing the two side by side across $(k, d, p)$, the Local Vector algorithm is ahead of or tied with QAOA at the shallowest depth $p=1$ for every $k$, while QAOA overtakes the Local Vector algorithm once the girth is large enough, leading to a quantum advantage for graph optimization at moderate girth.

Three natural questions remain open. The first concerns the ceiling against which all of these coefficients should be measured. For Max-Cut, Dembo, Montanari, and Sen \cite{Dembo_2017} showed that the optimal cut fraction on random $d$-regular graphs is $1/2 + P_*/\sqrt d + o(1/\sqrt d)$ with $P_* \approx 0.7632$ the Parisi value of the Sherrington--Kirkpatrick model, and Sen \cite{sen2017optimizationsparserandomhypergraphs} extended this Gaussian-optimization picture to Max-$k$-Cut, where the large-degree limit is governed by the ground-state energy of the corresponding mean-field Potts antiferromagnet, a Potts--Parisi variational problem. Computing this limit numerically for each $k$ would fix the optimal $1/\sqrt d$ coefficient $P^{\star}_{k}$ and turn our tables into an absolute comparison, quantifying how far both the Local Vector algorithm and QAOA sit below the optimum. 

The second open question is algorithmic. For the Sherrington--Kirkpatrick model, Montanari \cite{montanari2019optimizationsherringtonkirkpatrickhamiltonian} gave a message-passing algorithm that reaches the Parisi optimum whenever no overlap gap is present, and El Alaoui, Montanari, and Sellke \cite{alaoui2023localalgorithmsmaximumcut} constructed a local message-passing algorithm that achieves the near-optimal cut value on locally tree-like regular graphs for Max-Cut. It is natural to ask whether a message-passing scheme of this kind can be designed for Max-$k$-Cut and whether it provably attains the Potts--Parisi optimum $P^{\star}_k$, either matching or surpassing the guarantees established here. A positive answer would give the first classical algorithm provably optimal at leading order for Max-$k$-Cut. 

Finally, another open problem is how the Overlap Gap 
Property~\cite{overlap_gap_property,qaoa_needs_whole_graph_typical_case,limitations_local_quantum_algorithms} which obstructs the performance of QAOA, can be generalized beyond binary alphabets. Even an appropriate definition of the overlap gap property is missing in this setting. 

\section*{Acknowledgments}
R.S. thanks Ojas Parekh for inspiring discussions that helped frame the scope of this project.
The JPMorganChase team thanks Rob Otter for the executive support of the work and all of their colleagues at the Global Technology Applied Research center for helpful feedback and support throughout the project.

\paragraph{AI statement:} Large language models 
were used to develop the code for numerical simulation whose results are presented in \cref{sec:comparison}, to help with derivation of \cref{thm:beats_random}, and for editing the manuscript. 

\clearpage
\appendix
\section{Data for Cut Fraction Comparison}
\label{app:tables}

For completeness we record the numerical values underlying the comparison figures of \cref{sec:comparison}. \Cref{tab:asymptotic_coefficients} lists the infinite-degree $1/\sqrt d$ coefficients plotted in \cref{fig:infinite_comparison}.

\Cref{tab:comparison_k2,tab:comparison_k3,tab:comparison_k4} list the complete finite-degree cut fractions plotted in \cref{fig:finite_comparison}, one table per cut order $k$ and spanning every degree $d$ and depth $p$ computed. We also report the SDP as a further reference point, but with an important distinction. The Frieze--Jerrum and Goemans--Williamson guarantees are worst-case \emph{lower bounds} (\cref{subsec:sdp}). They therefore say nothing useful about the cut fraction actually attained on a typical instance. Instead, we compute the \emph{realized} SDP cut fraction by solving the SDP and rounding over $16$ independent random $d$-regular graphs on $n = 500$ vertices. Because interior-point SDP solving of this relaxation scales as $O(n^{3.5})$ in the number of vertices, this approach is limited to moderate graph sizes. As there is no dependence on $p$, we quote the value once per $k,d$ value.  

\begin{table}[htbp]
\centering
\caption{Coefficient of the $1/\sqrt d$ correction to the cut fraction above the random baseline $(k-1)/k$ in the infinite-degree limit, as a function of the QAOA/vector depth $p$ (girth $2p+2$, $m=p+1$). Columns are the Local Vector coefficient $C^{\star}_{m,k}=\sup_{\tau\geq0}C_{m,k}(\tau)$, the generalized Thompson--Parekh--Marwaha root-only coefficient $\alpha_k c_m$ ($\tau=0$), and the optimized QAOA coefficient $C^{\text{QAOA}}_{k,p}$. Each sub-table runs to the deepest depth reached by the QAOA optimization for that order $k$. At $p=1$ the improved rounding is unavailable, so the LV and TPM entries coincide. LV entries are subject to Monte Carlo error. These are the values plotted in \cref{fig:infinite_comparison}.}
\label{tab:asymptotic_coefficients}
\footnotesize
\setlength{\tabcolsep}{4pt}
\begin{minipage}[t]{0.33\textwidth}
\centering
\begin{tabular}[t]{c|ccc}
\toprule
\multicolumn{4}{c}{Max-2-Cut} \\
$p$ & LV & TPM & QAOA \\
\midrule
1 & 0.319 & 0.318 & 0.303 \\
2 & 0.503 & 0.450 & 0.408 \\
3 & 0.543 & 0.515 & 0.473 \\
4 & 0.569 & 0.551 & 0.516 \\
5 & 0.585 & 0.574 & 0.548 \\
6 & 0.597 & 0.588 & 0.572 \\
7 & 0.605 & 0.598 & 0.591 \\
8 & 0.611 & 0.605 & 0.607 \\
9 & 0.616 & 0.611 & 0.620 \\
10 & 0.619 & 0.615 & 0.631 \\
11 & 0.622 & 0.618 & 0.641 \\
12 & 0.625 & 0.621 & 0.649 \\
13 & 0.626 & 0.623 & 0.656 \\
14 & 0.628 & 0.624 & 0.662 \\
15 & 0.629 & 0.626 & 0.668 \\
16 & 0.630 & 0.627 & 0.673 \\
17 & 0.631 & 0.628 & 0.677 \\
18 & 0.632 & 0.629 & 0.681 \\
19 & 0.633 & 0.630 & 0.685 \\
20 & 0.633 & 0.630 & 0.688 \\
\bottomrule
\end{tabular}
\end{minipage}%
\begin{minipage}[t]{0.33\textwidth}
\centering
\begin{tabular}[t]{c|ccc}
\toprule
\multicolumn{4}{c}{Max-3-Cut} \\
$p$ & LV & TPM & QAOA \\
\midrule
1 & 0.358 & 0.358 & 0.350 \\
2 & 0.589 & 0.506 & 0.475 \\
3 & 0.638 & 0.579 & 0.552 \\
4 & 0.668 & 0.620 & 0.604 \\
5 & 0.688 & 0.645 & 0.642 \\
6 & 0.701 & 0.662 & 0.672 \\
7 & 0.711 & 0.673 & 0.696 \\
8 & 0.717 & 0.681 & 0.715 \\
9 & 0.723 & 0.687 & 0.732 \\
10 & 0.727 & 0.692 & 0.745 \\
11 & 0.730 & 0.695 & 0.757 \\
12 & 0.732 & 0.698 & 0.768 \\
13 & 0.734 & 0.701 & 0.777 \\
14 & 0.736 & 0.702 & 0.785 \\
15 & 0.737 & 0.704 & 0.792 \\
16 & 0.739 & 0.705 & 0.798 \\
17 & 0.740 & 0.706 & 0.804 \\
18 & 0.740 & 0.707 & 0.809 \\
19 & 0.741 & 0.708 & 0.814 \\
\bottomrule
\end{tabular}
\end{minipage}%
\begin{minipage}[t]{0.33\textwidth}
\centering
\begin{tabular}[t]{c|ccc}
\toprule
\multicolumn{4}{c}{Max-4-Cut} \\
$p$ & LV & TPM & QAOA \\
\midrule
1 & 0.353 & 0.353 & 0.353 \\
2 & 0.601 & 0.500 & 0.483 \\
3 & 0.650 & 0.571 & 0.562 \\
4 & 0.680 & 0.612 & 0.616 \\
5 & 0.699 & 0.636 & 0.656 \\
6 & 0.712 & 0.653 & 0.687 \\
7 & 0.722 & 0.664 & 0.712 \\
8 & 0.728 & 0.672 & 0.733 \\
9 & 0.733 & 0.678 & 0.750 \\
10 & 0.737 & 0.682 & 0.765 \\
11 & 0.740 & 0.686 & 0.777 \\
12 & 0.743 & 0.689 & 0.788 \\
13 & 0.745 & 0.691 & 0.798 \\
14 & 0.746 & 0.693 & 0.807 \\
15 & 0.748 & 0.694 & 0.814 \\
\bottomrule
\end{tabular}
\end{minipage}
\end{table}

\begin{sidewaystable}
\centering
\caption{Complete finite-degree cut fractions for Max-2-Cut on $d$-regular graphs of girth $\geq 2p+2$: Local Vector (LV, improved rounding), generalized Thompson--Parekh--Marwaha root-only rounding (TPM, $\tau=0$), and optimized QAOA, at every depth $p$ computed. The realized Frieze--Jerrum/GW SDP cut fraction (no $p$ dependence) is shown in each degree header. At $p=1$ the improved rounding is unavailable, so LV coincides with TPM; a dash (--) marks depths not computed for that method. LV and TPM entries are subject to Monte Carlo error.}
\label{tab:comparison_k2}
\scriptsize
\setlength{\tabcolsep}{4pt}
\begin{minipage}[t]{0.24\linewidth}
\centering
\begin{tabular}[t]{c|ccc}
\toprule
\multicolumn{4}{c}{$d=4$\ (\text{SDP}=0.850)} \\
\midrule
$p$ & LV & TPM & QAOA \\
\midrule
1 & 0.667 & 0.667 & 0.662 \\
2 & 0.755 & 0.730 & 0.716 \\
3 & 0.775 & 0.762 & 0.749 \\
4 & 0.789 & 0.781 & 0.769 \\
5 & 0.797 & 0.793 & 0.784 \\
6 & 0.805 & 0.801 & 0.795 \\
7 & 0.810 & 0.807 & 0.803 \\
8 & 0.814 & 0.812 & 0.810 \\
9 & 0.816 & 0.815 & 0.815 \\
10 & 0.819 & 0.818 & 0.820 \\
11 & 0.821 & 0.820 & 0.823 \\
12 & 0.822 & 0.821 & 0.826 \\
13 & 0.823 & 0.823 & 0.829 \\
14 & 0.824 & 0.824 & 0.832 \\
\midrule
\multicolumn{4}{c}{$d=6$\ (\text{SDP}=0.784)} \\
\midrule
$p$ & LV & TPM & QAOA \\
\midrule
1 & 0.634 & 0.634 & 0.629 \\
2 & 0.707 & 0.687 & 0.673 \\
3 & 0.723 & 0.713 & 0.699 \\
4 & 0.734 & 0.728 & 0.716 \\
5 & 0.742 & 0.738 & 0.729 \\
6 & 0.747 & 0.744 & 0.739 \\
7 & 0.751 & 0.749 & 0.746 \\
8 & 0.754 & 0.753 & 0.752 \\
9 & 0.756 & 0.755 & 0.757 \\
10 & 0.758 & 0.757 & 0.761 \\
11 & 0.759 & 0.758 & 0.764 \\
12 & 0.760 & 0.759 & 0.767 \\
13 & 0.761 & 0.761 & 0.769 \\
14 & 0.762 & 0.762 & 0.772 \\
\bottomrule
\end{tabular}
\end{minipage}%
\begin{minipage}[t]{0.24\linewidth}
\centering
\begin{tabular}[t]{c|ccc}
\toprule
\multicolumn{4}{c}{$d=8$\ (\text{SDP}=0.745)} \\
\midrule
$p$ & LV & TPM & QAOA \\
\midrule
1 & 0.615 & 0.615 & 0.611 \\
2 & 0.678 & 0.661 & 0.648 \\
3 & 0.693 & 0.684 & 0.671 \\
4 & 0.702 & 0.697 & 0.686 \\
5 & 0.709 & 0.705 & 0.697 \\
6 & 0.712 & 0.710 & 0.706 \\
7 & 0.716 & 0.714 & 0.712 \\
8 & 0.719 & 0.717 & 0.717 \\
9 & 0.720 & 0.719 & 0.722 \\
10 & 0.722 & 0.721 & 0.725 \\
11 & 0.723 & 0.723 & 0.728 \\
12 & 0.724 & 0.724 & 0.731 \\
13 & 0.725 & 0.724 & 0.733 \\
14 & 0.726 & 0.725 & 0.735 \\
\midrule
\multicolumn{4}{c}{$d=10$\ (\text{SDP}=0.718)} \\
\midrule
$p$ & LV & TPM & QAOA \\
\midrule
1 & 0.602 & 0.602 & 0.598 \\
2 & 0.660 & 0.643 & 0.632 \\
3 & 0.672 & 0.664 & 0.652 \\
4 & 0.681 & 0.676 & 0.666 \\
5 & 0.686 & 0.683 & 0.676 \\
6 & 0.690 & 0.688 & 0.683 \\
7 & 0.693 & 0.692 & 0.689 \\
8 & 0.695 & 0.694 & 0.694 \\
9 & 0.696 & 0.696 & 0.698 \\
10 & 0.698 & 0.697 & 0.701 \\
11 & 0.698 & 0.698 & 0.704 \\
12 & 0.699 & 0.699 & 0.706 \\
13 & 0.700 & 0.700 & 0.709 \\
14 & 0.701 & 0.700 & 0.710 \\
\bottomrule
\end{tabular}
\end{minipage}%
\begin{minipage}[t]{0.24\linewidth}
\centering
\begin{tabular}[t]{c|ccc}
\toprule
\multicolumn{4}{c}{$d=20$\ (\text{SDP}=0.653)} \\
\midrule
$p$ & LV & TPM & QAOA \\
\midrule
1 & 0.572 & 0.572 & 0.569 \\
2 & 0.612 & 0.601 & 0.592 \\
3 & 0.621 & 0.616 & 0.607 \\
4 & 0.627 & 0.624 & 0.616 \\
5 & 0.631 & 0.629 & 0.623 \\
6 & 0.634 & 0.632 & 0.629 \\
7 & 0.635 & 0.634 & 0.633 \\
8 & 0.637 & 0.636 & 0.636 \\
9 & 0.638 & 0.638 & 0.639 \\
10 & 0.638 & 0.638 & 0.642 \\
11 & 0.640 & 0.640 & 0.644 \\
12 & 0.640 & 0.640 & 0.646 \\
13 & 0.641 & 0.640 & 0.647 \\
14 & 0.641 & 0.641 & 0.648 \\
\midrule
\multicolumn{4}{c}{$d=30$\ (\text{SDP}=0.624)} \\
\midrule
$p$ & LV & TPM & QAOA \\
\midrule
1 & 0.558 & 0.558 & 0.556 \\
2 & 0.592 & 0.582 & 0.575 \\
3 & 0.599 & 0.594 & 0.587 \\
4 & 0.603 & 0.600 & 0.595 \\
5 & 0.606 & 0.605 & 0.600 \\
6 & 0.608 & 0.607 & 0.605 \\
7 & 0.611 & 0.610 & 0.608 \\
8 & 0.612 & 0.611 & 0.611 \\
9 & 0.613 & 0.613 & 0.614 \\
10 & 0.613 & 0.612 & 0.616 \\
11 & 0.614 & 0.614 & 0.617 \\
12 & 0.614 & 0.614 & 0.619 \\
13 & 0.614 & 0.614 & 0.620 \\
14 & 0.615 & 0.615 & 0.621 \\
\bottomrule
\end{tabular}
\end{minipage}%
\begin{minipage}[t]{0.24\linewidth}
\centering
\begin{tabular}[t]{c|ccc}
\toprule
\multicolumn{4}{c}{$d=40$\ (\text{SDP}=0.606)} \\
\midrule
$p$ & LV & TPM & QAOA \\
\midrule
1 & 0.550 & 0.550 & 0.548 \\
2 & 0.579 & 0.571 & 0.565 \\
3 & 0.586 & 0.582 & 0.575 \\
4 & 0.590 & 0.587 & 0.582 \\
5 & 0.592 & 0.591 & 0.587 \\
6 & 0.594 & 0.594 & 0.591 \\
7 & 0.595 & 0.594 & 0.594 \\
8 & 0.597 & 0.596 & 0.596 \\
9 & 0.598 & 0.598 & 0.598 \\
10 & 0.598 & 0.598 & 0.600 \\
11 & 0.598 & 0.598 & 0.601 \\
12 & 0.599 & 0.598 & 0.603 \\
13 & 0.599 & 0.599 & 0.603 \\
14 & 0.599 & 0.599 & 0.595 \\
\bottomrule
\end{tabular}
\end{minipage}%
\end{sidewaystable}

\begin{sidewaystable}
\centering
\caption{Complete finite-degree cut fractions for Max-3-Cut on $d$-regular graphs of girth $\geq 2p+2$: Local Vector (LV, improved rounding), generalized Thompson--Parekh--Marwaha root-only rounding (TPM, $\tau=0$), and optimized QAOA, at every depth $p$ computed. The realized Frieze--Jerrum/GW SDP cut fraction (no $p$ dependence) is shown in each degree header. At $p=1$ the improved rounding is unavailable, so LV coincides with TPM; a dash (--) marks depths not computed for that method. LV and TPM entries are subject to Monte Carlo error.}
\label{tab:comparison_k3}
\scriptsize
\setlength{\tabcolsep}{4pt}
\begin{minipage}[t]{0.24\linewidth}
\centering
\begin{tabular}[t]{c|ccc}
\toprule
\multicolumn{4}{c}{$d=4$\ (\text{SDP}=0.836)} \\
\midrule
$p$ & LV & TPM & QAOA \\
\midrule
1 & 0.836 & 0.836 & 0.833 \\
2 & 0.930 & 0.888 & 0.891 \\
3 & 0.944 & 0.912 & 0.921 \\
4 & 0.949 & 0.925 & 0.940 \\
5 & 0.954 & 0.932 & 0.952 \\
6 & 0.957 & 0.937 & 0.961 \\
7 & 0.960 & 0.941 & 0.968 \\
8 & 0.962 & 0.944 & 0.973 \\
9 & 0.963 & 0.945 & 0.977 \\
10 & 0.964 & 0.947 & 0.980 \\
\midrule
\multicolumn{4}{c}{$d=6$\ (\text{SDP}=0.836)} \\
\midrule
$p$ & LV & TPM & QAOA \\
\midrule
1 & 0.806 & 0.806 & 0.803 \\
2 & 0.888 & 0.853 & 0.851 \\
3 & 0.902 & 0.875 & 0.878 \\
4 & 0.909 & 0.887 & 0.895 \\
5 & 0.915 & 0.895 & 0.907 \\
6 & 0.918 & 0.899 & 0.917 \\
7 & 0.921 & 0.903 & 0.924 \\
8 & 0.923 & 0.905 & 0.930 \\
9 & 0.924 & 0.907 & 0.934 \\
10 & 0.925 & 0.908 & 0.938 \\
\bottomrule
\end{tabular}
\end{minipage}%
\begin{minipage}[t]{0.24\linewidth}
\centering
\begin{tabular}[t]{c|ccc}
\toprule
\multicolumn{4}{c}{$d=8$\ (\text{SDP}=0.836)} \\
\midrule
$p$ & LV & TPM & QAOA \\
\midrule
1 & 0.788 & 0.788 & 0.785 \\
2 & 0.861 & 0.831 & 0.827 \\
3 & 0.875 & 0.851 & 0.851 \\
4 & 0.882 & 0.862 & 0.867 \\
5 & 0.887 & 0.869 & 0.878 \\
6 & 0.890 & 0.873 & 0.887 \\
7 & 0.893 & 0.876 & 0.894 \\
8 & 0.895 & 0.879 & 0.899 \\
9 & 0.896 & 0.880 & 0.904 \\
10 & 0.897 & 0.882 & 0.908 \\
\midrule
\multicolumn{4}{c}{$d=10$\ (\text{SDP}=0.836)} \\
\midrule
$p$ & LV & TPM & QAOA \\
\midrule
1 & 0.776 & 0.776 & 0.773 \\
2 & 0.842 & 0.815 & 0.811 \\
3 & 0.854 & 0.833 & 0.832 \\
4 & 0.861 & 0.843 & 0.847 \\
5 & 0.866 & 0.850 & 0.858 \\
6 & 0.870 & 0.854 & 0.866 \\
7 & 0.872 & 0.857 & 0.872 \\
8 & 0.874 & 0.859 & 0.877 \\
9 & 0.875 & 0.861 & 0.881 \\
10 & 0.876 & 0.862 & 0.885 \\
\bottomrule
\end{tabular}
\end{minipage}%
\begin{minipage}[t]{0.24\linewidth}
\centering
\begin{tabular}[t]{c|ccc}
\toprule
\multicolumn{4}{c}{$d=20$\ (\text{SDP}=0.800)} \\
\midrule
$p$ & LV & TPM & QAOA \\
\midrule
1 & 0.744 & 0.744 & 0.743 \\
2 & 0.793 & 0.774 & 0.770 \\
3 & 0.803 & 0.788 & 0.786 \\
4 & 0.808 & 0.796 & 0.796 \\
5 & 0.812 & 0.801 & 0.804 \\
6 & 0.815 & 0.805 & 0.811 \\
7 & 0.817 & 0.807 & 0.815 \\
8 & 0.818 & 0.808 & 0.819 \\
9 & 0.819 & 0.810 & 0.823 \\
10 & 0.820 & 0.810 & 0.825 \\
\midrule
\multicolumn{4}{c}{$d=30$\ (\text{SDP}=0.776)} \\
\midrule
$p$ & LV & TPM & QAOA \\
\midrule
1 & 0.730 & 0.730 & 0.729 \\
2 & 0.771 & 0.755 & 0.751 \\
3 & 0.779 & 0.767 & 0.764 \\
4 & 0.784 & 0.774 & 0.773 \\
5 & 0.787 & 0.778 & 0.780 \\
6 & 0.789 & 0.781 & 0.785 \\
7 & 0.791 & 0.782 & 0.789 \\
8 & 0.792 & 0.784 & 0.792 \\
9 & 0.793 & 0.785 & 0.795 \\
10 & 0.793 & 0.786 & 0.798 \\
\bottomrule
\end{tabular}
\end{minipage}%
\begin{minipage}[t]{0.24\linewidth}
\centering
\begin{tabular}[t]{c|ccc}
\toprule
\multicolumn{4}{c}{$d=40$\ (\text{SDP}=0.762)} \\
\midrule
$p$ & LV & TPM & QAOA \\
\midrule
1 & 0.722 & 0.722 & 0.721 \\
2 & 0.757 & 0.744 & 0.740 \\
3 & 0.764 & 0.754 & 0.752 \\
4 & 0.769 & 0.760 & 0.760 \\
5 & 0.771 & 0.764 & 0.765 \\
6 & 0.773 & 0.766 & 0.770 \\
7 & 0.775 & 0.768 & 0.773 \\
8 & 0.776 & 0.769 & 0.776 \\
9 & 0.776 & 0.770 & 0.779 \\
10 & 0.777 & 0.771 & 0.791 \\
\bottomrule
\end{tabular}
\end{minipage}%
\end{sidewaystable}

\begin{sidewaystable}
\centering
\caption{Complete finite-degree cut fractions for Max-4-Cut on $d$-regular graphs of girth $\geq 2p+2$: Local Vector (LV, improved rounding), generalized Thompson--Parekh--Marwaha root-only rounding (TPM, $\tau=0$), and optimized QAOA, at every depth $p$ computed. The realized Frieze--Jerrum/GW SDP cut fraction (no $p$ dependence) is shown in each degree header. At $p=1$ the improved rounding is unavailable, so LV coincides with TPM; a dash (--) marks depths not computed for that method. LV and TPM entries are subject to Monte Carlo error.}
\label{tab:comparison_k4}
\scriptsize
\setlength{\tabcolsep}{4pt}
\begin{minipage}[t]{0.24\linewidth}
\centering
\begin{tabular}[t]{c|ccc}
\toprule
\multicolumn{4}{c}{$d=4$\ (\text{SDP}=0.858)} \\
\midrule
$p$ & LV & TPM & QAOA \\
\midrule
1 & 0.904 & 0.904 & 0.901 \\
2 & 0.980 & 0.944 & 0.953 \\
3 & 0.986 & 0.960 & 0.977 \\
4 & 0.989 & 0.968 & 0.988 \\
5 & 0.991 & 0.973 & 0.993 \\
6 & 0.990 & 0.976 & 0.996 \\
7 & 0.991 & 0.978 & 0.998 \\
8 & 0.992 & 0.979 & 0.999 \\
9 & 0.992 & 0.980 & -- \\
10 & 0.993 & 0.981 & -- \\
\midrule
\multicolumn{4}{c}{$d=6$\ (\text{SDP}=0.858)} \\
\midrule
$p$ & LV & TPM & QAOA \\
\midrule
1 & 0.879 & 0.879 & 0.877 \\
2 & 0.953 & 0.918 & 0.923 \\
3 & 0.963 & 0.934 & 0.946 \\
4 & 0.968 & 0.943 & 0.961 \\
5 & 0.971 & 0.948 & 0.970 \\
6 & 0.971 & 0.951 & 0.977 \\
7 & 0.972 & 0.954 & 0.982 \\
8 & 0.974 & 0.956 & 0.985 \\
9 & 0.974 & 0.957 & -- \\
10 & 0.975 & 0.958 & -- \\
\bottomrule
\end{tabular}
\end{minipage}%
\begin{minipage}[t]{0.24\linewidth}
\centering
\begin{tabular}[t]{c|ccc}
\toprule
\multicolumn{4}{c}{$d=8$\ (\text{SDP}=0.858)} \\
\midrule
$p$ & LV & TPM & QAOA \\
\midrule
1 & 0.864 & 0.864 & 0.862 \\
2 & 0.933 & 0.900 & 0.903 \\
3 & 0.943 & 0.916 & 0.925 \\
4 & 0.949 & 0.924 & 0.939 \\
5 & 0.952 & 0.930 & 0.949 \\
6 & 0.953 & 0.933 & 0.956 \\
7 & 0.955 & 0.935 & 0.961 \\
8 & 0.956 & 0.937 & 0.966 \\
9 & 0.956 & 0.938 & -- \\
10 & 0.957 & 0.939 & -- \\
\midrule
\multicolumn{4}{c}{$d=10$\ (\text{SDP}=0.858)} \\
\midrule
$p$ & LV & TPM & QAOA \\
\midrule
1 & 0.853 & 0.853 & 0.851 \\
2 & 0.917 & 0.886 & 0.888 \\
3 & 0.927 & 0.902 & 0.909 \\
4 & 0.933 & 0.910 & 0.922 \\
5 & 0.937 & 0.915 & 0.932 \\
6 & 0.938 & 0.919 & 0.939 \\
7 & 0.939 & 0.921 & 0.944 \\
8 & 0.941 & 0.922 & 0.949 \\
9 & 0.942 & 0.924 & -- \\
10 & 0.942 & 0.925 & -- \\
\bottomrule
\end{tabular}
\end{minipage}%
\begin{minipage}[t]{0.24\linewidth}
\centering
\begin{tabular}[t]{c|ccc}
\toprule
\multicolumn{4}{c}{$d=20$\ (\text{SDP}=0.858)} \\
\midrule
$p$ & LV & TPM & QAOA \\
\midrule
1 & 0.825 & 0.825 & 0.824 \\
2 & 0.873 & 0.851 & 0.851 \\
3 & 0.882 & 0.864 & 0.866 \\
4 & 0.887 & 0.870 & 0.877 \\
5 & 0.891 & 0.875 & 0.885 \\
6 & 0.892 & 0.878 & 0.891 \\
7 & 0.894 & 0.880 & 0.895 \\
8 & 0.895 & 0.881 & 0.899 \\
9 & 0.896 & 0.882 & -- \\
10 & 0.896 & 0.883 & -- \\
\midrule
\multicolumn{4}{c}{$d=30$\ (\text{SDP}=0.847)} \\
\midrule
$p$ & LV & TPM & QAOA \\
\midrule
1 & 0.811 & 0.811 & 0.811 \\
2 & 0.852 & 0.834 & 0.833 \\
3 & 0.860 & 0.845 & 0.847 \\
4 & 0.864 & 0.851 & 0.855 \\
5 & 0.867 & 0.854 & 0.862 \\
6 & 0.869 & 0.857 & 0.867 \\
7 & 0.870 & 0.858 & 0.871 \\
8 & 0.871 & 0.860 & 0.874 \\
9 & 0.872 & 0.861 & -- \\
10 & 0.872 & 0.861 & -- \\
\bottomrule
\end{tabular}
\end{minipage}%
\begin{minipage}[t]{0.24\linewidth}
\centering
\begin{tabular}[t]{c|ccc}
\toprule
\multicolumn{4}{c}{$d=40$\ (\text{SDP}=0.835)} \\
\midrule
$p$ & LV & TPM & QAOA \\
\midrule
1 & 0.804 & 0.804 & 0.803 \\
2 & 0.840 & 0.824 & 0.823 \\
3 & 0.846 & 0.833 & 0.834 \\
4 & 0.851 & 0.839 & 0.842 \\
5 & 0.853 & 0.842 & 0.848 \\
6 & 0.854 & 0.844 & 0.852 \\
7 & 0.856 & 0.846 & 0.856 \\
8 & 0.856 & 0.847 & 0.859 \\
9 & 0.857 & 0.848 & -- \\
10 & 0.858 & 0.848 & -- \\
\bottomrule
\end{tabular}
\end{minipage}%
\end{sidewaystable}

\clearpage
\printbibliography

\part*{Disclaimer}

This paper was prepared for informational purposes with contributions from the Global Technology Applied Research center of JPMorgan Chase \& Co. This paper is not a product of the Research Department of JPMorgan Chase \& Co. or its affiliates. Neither JPMorgan Chase \& Co. nor any of its affiliates makes any explicit or implied representation or warranty and none of them accept any liability in connection with this paper, including, without limitation, with respect to the completeness, accuracy, or reliability of the information contained herein and the potential legal, compliance, tax, or accounting effects thereof. This document is not intended as investment research or investment advice, or as a recommendation, offer, or solicitation for the purchase or sale of any security, financial instrument, financial product or service, or to be used in any way for evaluating the merits of participating in any transaction.

\end{document}